\documentclass[pdflatex,sn-mathphys-num]{sn-jnl}

\usepackage{graphicx}%
\usepackage{multirow}%
\usepackage{amsmath,amssymb,amsfonts}%
\usepackage{amsthm}%
\usepackage[title]{appendix}%
\usepackage{xcolor}%
\usepackage{textcomp}%
\usepackage{manyfoot}%
\usepackage{booktabs}%
\usepackage{algorithm}%
\usepackage{algorithmicx}%
\usepackage[noend]{algpseudocode}%
\usepackage{listings}%

\usepackage{amsmath}
\usepackage{graphicx}
\usepackage{amsfonts}
\usepackage{amssymb}
\usepackage{float}
\usepackage{booktabs}
\usepackage{enumerate}
\usepackage{placeins}  
\usepackage{breakcites}

\newcommand{\SSR}{\ensuremath{\textsc{Subset Sum Ratio}}}
\newcommand{\kSSR}{\ensuremath{k\textsc{-Subset Sum Ratio}}}

\newcommand{\SubS}{\ensuremath{\textsc{Subset Sum}}}
\newcommand{\ESS}{\textsc{Equal Subset Sum}}

\DeclareMathOperator*{\argmax}{arg\,max}
\DeclareMathOperator*{\argmin}{arg\,min}

\newcommand{\mP}{\ensuremath{\allowbreak \textsc{Multiway}\allowbreak \ \textsc{Number}\  \allowbreak \textsc{Partitioning}}}

\newcommand{\kPR}{\ensuremath{k\textsc{-way Number Partitioning Ratio}}}
\newcommand{\kPRshort}{\ensuremath{k\text{-}\mathrm{PART}}}
\newcommand{\kSSRshort}{\ensuremath{ k\text{-}\mathrm{SSR}}}
\newcommand{\kPRRshort}{\ensuremath{k\text{-}\mathrm{PART_R}}}
\newcommand{\kSSRLshort}{\ensuremath{k\text{-}\mathrm{SSR_L}}}
\newcommand{\kSSRRshort}{\ensuremath{k\text{-}\mathrm{SSR_R}}}

\theoremstyle{thmstyleone}%
\newtheorem{theorem}{Theorem}
\newtheorem{definition}[theorem]{Definition}
\newtheorem{observation}[theorem]{Observation}
\newtheorem{lemma}[theorem]{Lemma}

\begin{document}

\title[Approximation Schemes for k-Subset Sum Ratio and k-way Number Partitioning Ratio]{Approximation Schemes for k-Subset Sum Ratio and k-way Number Partitioning Ratio}











\author*[1,2]{\fnm{Sotiris} \sur{Kanellopoulos}
}\email{s.kanellopoulos@athenarc.gr}

\author*[1,2,3]{\fnm{Giorgos} \sur{Mitropoulos}}\email{georgios.mitropoulos@lip6.fr}

\author[1]{\fnm{Antonis} \sur{Antonopoulos}}

\author[1]{\fnm{Nikos} \sur{Leonardos}}

\author[1,2]{\fnm{Aris} \sur{Pagourtzis}}

\author[1,2]{\fnm{Christos} \sur{Pergaminelis}}

\author[1]{\fnm{Stavros} \sur{Petsalakis}}

\author[1]{\fnm{Kanellos} \sur{Tsitouras}}

\affil[1]{
\orgname{National Technical University of Athens}, \orgaddress{
\country{Greece}}}

\affil[2]{\orgdiv{Archimedes, \orgname{Athena Research Center}, \orgaddress{
\country{Greece}}}}

\affil[3]{\orgdiv{LIP6}, \orgname{Sorbonne Université, CNRS}, \orgaddress{
\city{Paris}, \postcode{F-75005}, 
\country{France}}}

\abstract{The \textsc{Subset Sum Ratio} problem (SSR) asks, given a multiset $A$ of $n$ positive integers, to find two disjoint subsets of $A$ such that the \emph{largest-to-smallest} ratio of their sums is minimized.
In this paper we study the $k$-version of SSR, namely $k$-\textsc{Subset Sum Ratio} ($k$-SSR), which asks to minimize the \emph{largest-to-smallest} ratio of sums of $k$ disjoint subsets of $A$. We develop approximation schemes for $k$-SSR running in $O({n^{2k}}/{\varepsilon^{k-1}})$ and $\widetilde{O}(n/{\varepsilon^{3k-1}})$ time, by expanding known techniques for SSR and combining them with novel ideas to address the multi-subset setting. Our second FPTAS employs carefully designed calls to the first one, being much faster than the first one when $n\gg 1/\varepsilon$. To the best of our knowledge, these are the first FPTASs for $k$-SSR for fixed $k>2$.

We also study the variant of \textsc{Multiway Number Partitioning} in which the objective is to minimize the largest-to-smallest sum ratio of $k$ subsets. We extend our first FPTAS for $k$-SSR to this problem through additional arguments to comply with the partitioning constraint, also achieving $O({n^{2k}}/{\varepsilon^{k-1}})$ time complexity. This represents an improvement over the best known bound of $O(n^{4k^2+1}/\varepsilon^{2k^2})$ for this problem, which stems from Nguyen and Rothe's FPTAS~\cite{FD_NR} for \textsc{Minimum Envy-Ratio}, although that FPTAS concerns a more general problem.


}




\maketitle

\section{Introduction}

The \textsc{Subset Sum} and \textsc{Partition} problems are fundamental in combinatorial optimization, with numerous applications in fair resource allocation, cryptography and scheduling. The $\ESS$ (ESS) problem~\cite{SSR_INTRODUCED} asks for two disjoint subsets of a given set with equal sums, while its optimization counterpart, $\SSR$ (SSR)~\cite{SSR_INTRODUCED,SSR_Baz}, asks to minimize the ratio of sums of two disjoint subsets. These NP-hard problems arise in bioinformatics~\cite{PD_Ciel2003,PD_Ciel2004}, game theory and economics~\cite{FD_LM, DFR2025}, and cryptography~\cite{ZVVH24}, among others. A special case of ESS belongs to PPP~\cite{PPP_Pap}, which is a subclass of TFNP, further adding to the ongoing interest for these problems.

In this paper we consider variations involving multiple subsets, namely $k$-\textsc{Subset Sum Ratio} ($\kSSRshort$) and $k$-\textsc{way Number Partitioning Ratio} ($k$-\textsc{PART})\footnote{This problem is usually referred to as "multiway", but we call it "$k$-way" to emphasize that $k$ is fixed.}. In $\kSSRshort$, the goal is to find $k$ disjoint subsets such that the \emph{largest-to-smallest} ratio of their sums is minimized,
 while $\kPRshort$ further requires these subsets to form a partition of the input. Since these problems are NP-hard, we consider approximation algorithms. In particular, we present fully polynomial-time approximation schemes (FPTASs) for both problems.

 Our results apply, among others, to scheduling tasks across $k$ processors (e.g.,~\cite{MUL_Sah,SCH_Br,HOCH_APPROX,kVisits}), as well as to the \emph{fair division of indivisible goods}, in particular the study of \emph{envy} in discrete settings (e.g.,~\cite{FD_LM,FD_NR,FD_KPW} and more recently \cite{EFX_Amanat,EFX_Haj,EFX_Zhou,EFX_param}). Specifically, $\kPRshort$ is equivalent to the \textsc{Minimum Envy-Ratio} problem~\cite{FD_LM} with identical valuation functions, for which we achieve a significant improvement over (more general) earlier work~\cite{FD_NR}.

\paragraph*{Related work} 
Since Bellman’s seminal work on \textsc{Subset Sum}~\cite{Bel}, no major improvement had occurred until a recent spike in interest culminated in various faster pseudo-polynomial algorithms~\cite{SS_KX,SS_Bri,SS_JW,SS_CLMZ} and FPTASs \cite{SS_Kel,SS_Gens}.

The ESS and SSR problems have seen recent advances, including FPTASs for SSR~\cite{SSR_Nan,SSR_MP,SSR_Alo,SSR_Bri} and an improved exact algorithm for ESS~\cite{ESS_Muc}. Notably, Bringmann~\cite{SSR_Bri} provided an SSR FPTAS running in time $O(n/\varepsilon^{0.9386})$, which is faster than known \textsc{Subset Sum} lower bounds. He achieved this in part by considering two cases based on input density. If the input is dense, a pigeonhole argument guarantees two subsets with sums within $1+\varepsilon$ and removing their intersection yields an approximate solution. If the input is sparse, keeping only the $\mathrm{polylog}(1/\varepsilon)$ largest elements is proven to be sufficient for the approximation.
Regarding variations of (\textsc{Equal}) $\SubS$ involving $k\geq2$ subsets, Antonopoulos et al.~\cite{KSS} extend the techniques of~\cite{SS_KX,SS_Bri} to provide both deterministic and randomized pseudo-polynomial algorithms.

For \textsc{Partition}, advances include a number of (very) recent FPTASs \cite{PAR_Muc,PAR_BN,PAR_Lin}, with the latest work by Chen et al. \cite{PAR_NEW} achieving an FPTAS running in time $\widetilde{O}(n + 1/\varepsilon)$, which is near optimal under the Strong Exponential Time Hypothesis. In $\mP$~\cite{MUL_Bis,MUL_SKM,MUL_Kor2,MUL_Rec}, different objective functions define various approaches to distributing elements among subsets, such as minimizing the maximum sum, maximizing the minimum sum, minimizing the \emph{largest-to-smallest} difference of sums, or minimizing the \emph{largest-to-smallest} ratio of sums. 
These objectives are not comparable: there are instances showing that the corresponding optimal solutions differ; we discuss this in detail in  Subsection~\ref{subsec_objfunc}, along with applications of each variation.

Bazgan, Santha and Tuza~\cite{SSR_Baz} provided the first FPTAS for SSR and demonstrated that \textsc{Subset Sum Difference} does not admit an FPTAS, unless $\mathsf{P} = \mathsf{NP}$. Similar arguments that apply to \textsc{Partition} have been further studied in \cite{FPTAS_Woe}.
Additionally, various related problems such as \textsc{3-Partition} and \textsc{Bin Packing} are strongly NP-hard~\cite{BOOK_Gar}, implying that they do not admit an FPTAS or a pseudo-polynomial time algorithm unless $\mathsf{P}=\mathsf{NP}$. However, these problems involve a non-constant number of subsets. For our problems, we consider a \emph{fixed} number $k$ of disjoint subsets. This restriction is quite natural in applications, as it is more likely to distribute a great amount of items to a small number of agents.

Sahni presents an FPTAS for $\mP$ in \cite{MUL_Sah} that minimizes the largest sum. However, this technique does not seem applicable to minimize the ratio.
Lipton, Markakis, Mossel and Saberi~\cite{FD_LM} show a PTAS for \textsc{Minimum Envy-Ratio} with identical valuations and non-constant number $k$ of agents, while claiming the existence of an FPTAS when $k$ is constant. 
Furthermore, Nguyen and Rothe~\cite{FD_NR} show an FPTAS running in $O(n^{4k^2+1}/\varepsilon^{2k^2})$ time for the same problem with fixed $k$ and general additive valuation functions.

\paragraph*{Our contribution}

In this paper, we present two FPTASs for $\kSSRshort$ achieving running times\footnote{$\widetilde{O}$ hides $\mathrm{polylog}(1/\varepsilon)$ factors.} of $O(n^{2k}/\varepsilon^{k-1})$ and $\widetilde{O}(n/{\varepsilon^{3k-1}})$, and an FPTAS for $\kPRshort$ running in time $O(n^{2k}/\varepsilon^{k-1})$. The latter represents a considerable improvement when compared to the FPTAS of~\cite{FD_NR}, which is the best known bound for $\kPRshort$, although it concerns the more general setting of \emph{non-identical additive} valuations. Note that the largest-to-smallest ratio of sums is the only objective function considered in the literature that both admits an FPTAS and is of interest to the study of \emph{envy}. Our approach builds upon methods established for SSR~\cite{SSR_Nan,SSR_MP,SSR_Bri}, while introducing novel strategies to address the multi-subset setting and to comply with the partitioning constraint.

We present our first FPTAS for $\kSSRshort$ in Section~\ref{sec:FPTAS_kSSR}; to the best of our knowledge, this is the first ever FPTAS for this problem. To this end, we define a restricted version of $\kSSRshort$ (analogous to that of Melissinos and Pagourtzis~\cite{SSR_MP} for SSR), for which we present a pseudo-polynomial time algorithm. Our approach relies on proving that the optimal solution of the restricted problem satisfies certain properties, while carefully handling edge cases involving large singleton sets.

We then extend our first FPTAS to $\kPRshort$ in Section~\ref{sec:FPTAS_kPR}.
The partitioning constraint renders the application of our techniques more involved; we overcome this by identifying a \emph{perfect} restriction parameter that ensures the existence of a \emph{well-behaved} optimal solution. Interestingly, although this does not yield a pseudo-polynomial time algorithm for the respective restricted case, combining the aforementioned property with standard approximation techniques yields an FPTAS for $\kPRshort$.

Lastly, building on Bringmann's work \cite{SSR_Bri}, we find that applying density arguments to $\kSSRshort$ encounters an obstacle in the dense case: obtaining $k$ subsets with approximately equal sums is unhelpful, since removing their intersection does not guarantee a feasible $\kSSRshort$ solution. We resolve this in Section~\ref{sec:FPTAS_kSSR_alternative} by restricting our density argument to singletons and subsequently using our FPTAS from Section~\ref{sec:FPTAS_kSSR} as a subroutine on reduced instances. 

\paragraph*{Prior work}

This work is an extension of the conference version presented at ISAAC 2025~\cite{ISAAC_version}. Compared to that, this version contains the proof of Theorem~\ref{thrm:approximation_ratio} (Subsection~\ref{subsec:proof_new}), which establishes that our first FPTAS for \kSSRshort\ is a $(1+\varepsilon)$-approximation. This is also necessary to bound the approximation ratio of our other two FPTASs.
Additionally, we provide the omitted pseudocodes that are necessary for understanding our algorithms in depth and make several revisions to improve clarity.
Lastly, we include an in-depth discussion of previous work on SSR and a comparison with our results (Subsection~\ref{subsec:previous_work_SSR}), as well as a conclusion (Section~\ref{sec:conclusion}) discussing several open questions and directions for future work.



\section{Preliminaries and problem definitions}\label{sec2}

\subsection{Overview of existing FPTASs for Subset Sum Ratio}
\label{subsec:previous_work_SSR}

In this subsection we briefly discuss the existing FPTASs for SSR ($k=2$) and compare them to our own FPTASs for \kSSRshort.

Table~\ref{tab:SSR} summarizes the running times of known FPTASs for SSR. The FPTAS of Nanongkai~\cite{SSR_Nan} is omitted, as its focus was providing a simpler algorithm compared to~\cite{SSR_Baz}, without improving over the running time. Note that all of these FPTASs are only known to work for $2$ subsets; none of them seem to generalize directly for $k>2$ without major modifications.

\renewcommand{\arraystretch}{1.3}
\begin{table}[h]
    \centering
    \begin{tabular}{|c|c|}
        \hline
        \textbf{Paper} & \textbf{FPTAS time} \\
        \hline
         Bazgan et al. '02~\cite{SSR_Baz} & $O(n^5/\varepsilon^3)$ \\
        \hline
         Melissinos, Pagourtzis '18~\cite{SSR_MP} & $O(n^4/\varepsilon)$ \\
        \hline
         Alonistiotis et al. '24~\cite{SSR_Alo} & $\widetilde{O}(n^{2.3}/\varepsilon^{2.6})$ \\
         & $\widetilde{O}(n^2/\varepsilon^3)$ \\
        \hline
         Bringmann '24~\cite{SSR_Bri} & $\widetilde{O}(n/\varepsilon)$ \\
         & $O(n/\varepsilon^{0.9386})$ \\
        \hline
    \end{tabular}
    \caption{History of FPTASs for SSR.}
    \label{tab:SSR}
\end{table}

The FPTASs of~\cite{SSR_Baz,SSR_Nan,SSR_MP} all rely on dynamic programming (DP). The FPTASs of Alonistiotis et al.~\cite{SSR_Alo} employ calls to exact and approximate algorithms for \textsc{Partition}, while the FPTASs of Bringmann~\cite{SSR_Bri} utilize density arguments based on obtaining approximately equal sums via the pigeonhole principle on rounded instances. Note that the first FPTAS of Bringmann~\cite{SSR_Bri} uses the FPTAS of Melissinos and Pagourtzis~\cite{SSR_MP} as a subroutine on reduced instances.

Among the FPTASs that are based on DP, the fastest is the one by Melissinos and Pagourtzis~\cite{SSR_MP}. However, its techniques are limited to the $k=2$ case: for example, the way the authors handle the validity of restricted solutions during the DP is specifically designed for two subsets, while the $k=2$ case also allows them to treat large singletons as a trivial edge case. In Section~\ref{sec:FPTAS_kSSR}, we expand on the techniques of~\cite{SSR_Nan, SSR_MP} by using a simple Boolean vector to handle the validity of restricted solutions during DP and by carefully crafting a \emph{well-behavedness} theorem (Theorem~\ref{thrm:optimal_sol}) that dictates how to handle large singletons. Conflict resolution in DP becomes significantly more complicated compared to the $k=2$ version, due to the existence of large singletons not participating in DP (Lemma~\ref{lem:DP_conflicts}). Regardless of these obstacles, our slowdown compared to~\cite{SSR_MP} only includes functions of $k$; in particular, our asymptotic complexity $\bigl(O({n^{2k}}/{\varepsilon^{k-1}})\bigr)$ matches that of~\cite{SSR_MP} for $k=2$.

Unfortunately, the ideas leading to the current state-of-the-art FPTAS for SSR by Bringmann~\cite{SSR_Bri} do not seem applicable for $k>2$ subsets. Although one could use a generalized pigeonhole argument on a rounded instance to obtain $k$ subsets with approximately equal sums, the pigeonhole principle does not guarantee \emph{disjointness}, while for $k=2$ one can simply remove the intersection to obtain disjoint subsets (cf.~\cite{SSR_Bri}).
In Section~\ref{sec:FPTAS_kSSR_alternative}, we overcome this obstacle by using a density argument restricted to singletons, guaranteeing disjointness. Using our first FPTAS as subroutine on reduced instances, we still manage to obtain an FPTAS for \kSSRshort\ running in time linear in $n$ $\bigl(\widetilde{O}(n/{\varepsilon^{3k-1}})\bigr)$, although its dependency on $\varepsilon$ is worse than Bringmann's FPTAS, since our density argument for $k>2$ is weaker.

\subsection{Objective functions for Multiway Number Partitioning}
\label{subsec_objfunc}

There are four objective functions commonly used for $\mP$, each one having applications in different fields.
\begin{enumerate}
    \item Minimizing the largest sum. This is also known as makespan minimization and is used for the Minimum Finish Time problem in the context of identical machine scheduling \cite{MUL_Sah,HOCH_APPROX}. It is also used in problems related to bin packing \cite{BIN_Leu}. \label{objfunc1} 
    \item Maximizing the smallest sum. This objective function has been studied in works related to scheduling \cite{SCH_Br} and bin packing \cite{BIN_Leu}. It has also been studied in works related to the \emph{maximin share} \cite{FD_TH,FD_KPW,FD_EB}, which is a criterion for fair item allocation. Note that algorithms for objective function (\ref{objfunc2}) can be modified to work for (\ref{objfunc1}) and vice versa \cite{MUL_Kor2}. 
    \label{objfunc2}
    \item Minimizing the difference between the largest and the smallest sum. This objective function is less common, but it has been studied in works related to $\mP$, such as \cite{MUL_SM,MUL_Kor2}. No FPTAS exists for this objective~\cite{SSR_Baz,FPTAS_Woe}. \label{objfunc3}
    \item Minimizing the ratio between the largest and the smallest sum (e.g. \cite{KSS}). Although this objective function is often overlooked in works regarding $\mP$ (e.g. \cite{MUL_Kor2,MUL_SKM}), it has been studied in the field of fair division, in the form of the \textsc{Minimum Envy-Ratio} problem \cite{FD_LM,FD_NR}, as well as in the context of scheduling~\cite{SCH_Coff}. \label{objfunc4}
\end{enumerate}

In this paper, we consider (\ref{objfunc4}) as an objective function for $k$-way Number Partitioning. Observe that all four objectives are equivalent when $k=2$; this does not hold for $k\geq 3$. This is well known for the first three objectives \cite{MUL_Kor2}, but, for the sake of completeness, we will provide some counterexamples that prove that (\ref{objfunc4}) is not equivalent to any of the other three, for $k\geq 3$.

It is easy to find counterexamples for the first two. Consider input $\{1,2,3,10\}$ for $k=3$. The partition $(\{1 \},\{2,3 \},\{10 \})$ minimizes the largest sum, but does not have optimal ratio. Similarly, consider $\{5,5,5,10\}$ for $k=3$. The partition $(\{5 \},\{5 \},\{5, 10 \})$ maximizes the smallest sum, but does not have optimal ratio. 

Finally, consider input $\{16,16,18,20,24,27,29,40 \}$ for $k=4$. The following table compares the solution that minimizes the difference with the solution that minimizes the ratio.

\renewcommand{\arraystretch}{1}
\begin{table}[h]
    \centering
    \begin{tabular}{|c|c|c|c|}
        \hline
        \textbf{Solution} & \textbf{Sums} & \textbf{Difference} & \textbf{Ratio} \\
        \hline
        \{40\}, \{16, 16, 18\}, \{24, 27\}, \{29, 20\}  & \{40,50,51,49\} &\ \textbf{11} & 1.275 \\
        \hline
        \{40, 16\}, \{24, 20\}, \{16, 29\}, \{18, 27\}  & \{56,44,45,45\} &\ 12 &\ \textbf{1.27273} \\
        \hline
    \end{tabular}
    \caption{Comparison of $\mP$ solutions.}
    \label{tab:partition_results}
\end{table}

\subsection{Notation and problem definitions}

Let $[n]=\{1,\dots, n\}$ for a positive integer $n$. 
We assume that the input $A=\{a_1,\ldots,a_n\}$ of our problems is sorted\footnote{If the input is not sorted, an additional $O(n\log n)$ time would be required. This does not affect the time complexity of any of our algorithms, assuming that $\widetilde{O}$ also hides $\mathrm{polylog}(n)$ factors.}, i.e. $a_1 \leq \ldots \leq a_n$. 

\begin{definition}[Sum of a subset of indices]
\label{def:index_sum}
Given a multiset $A=\{a_1,\ldots,a_n\}$ of positive integers and a set of indices $S \subseteq [n]$ we define:
$$\Sigma(S,A)=\sum_{i \in S} a_i$$
\end{definition}

\begin{definition}[Max and Min of $k$ subsets]
\label{def:max_min}
Let $A=\{a_1,\ldots,a_n\}$ be a multiset of positive integers and $S_1,\ldots,S_k$ be $k$ disjoint subsets of $[n]$. We define the maximum and minimum sums obtained by these sets on $A$ as:
$$M(S_1,\ldots,S_k,A) = \max_{i \in [k]}\Sigma(S_i,A) \quad \textnormal{and} \quad m(S_1,\ldots,S_k,A) = \min_{i \in [k]}\Sigma(S_i,A)$$
\end{definition}

\begin{definition}[Largest-to-smallest ratio of $k$ subsets]
\label{def:k-ratio}
Let $A=\{a_1,\ldots,a_n\}$ be a multiset of positive integers and $S_1,\ldots,S_k$ be $k$ disjoint subsets of $[n]$. We define the \emph{largest-to-smallest ratio} of these $k$ subsets on $A$ as:
$$\mathcal{R}(S_1,\ldots,S_k,A)=
\begin{cases}
\frac{M(S_1,\ldots,S_k,A)}{m(S_1,\ldots,S_k,A)} & \text{if } m(S_1,\ldots,S_k,A)>0\\
+\infty  & \text{if } m(S_1,\ldots,S_k,A)=0
\end{cases}$$
\end{definition}
Throughout the paper, we refer to the largest-to-smallest ratio of $k$ subsets simply as \textit{ratio}. 

We define the $\kSSR$ ($\kSSRshort$) problem as follows.

\begin{definition}[$\kSSRshort$]
\label{def:kSSR}
    Given a multiset $A=\{a_1,\ldots,a_n\}$ of positive integers, find $k$ disjoint subsets $S_1,\ldots,S_k$ of $[n]$, such that $\mathcal{R}(S_1,\ldots,S_k,A)$ is minimized.
\end{definition}

\begin{observation}
\label{obs:ratio_change}
Only the maximum and minimum sums affect the ratio function. If the remaining sets are altered without a sum becoming greater than the maximum or smaller than the minimum, the ratio remains unaffected. Additionally, if the minimum sum increases or the maximum sum decreases (while the other remains unchanged), the ratio decreases.
\end{observation}

Note that multisets are allowed as input, since they are not a trivial case for $\kSSRshort$ (in contrast to regular SSR), unless there is a number with multiplicity $k$ or more. Throughout the paper, when referring to a solution $S=(S_1,\ldots,S_k)$, we will use the simplified notations $\mathcal{R}(S,A),M(S,A),m(S,A)$ to denote ratio, maximum and minimum sums respectively.
\smallskip

We similarly define the $\kPR$ ($\kPRshort$) problem.

\begin{definition}[$\kPRshort$]
\label{def:kPR}
    Given a multiset $A=\{a_1,\ldots,a_n\}$ of positive integers, find $k$ disjoint subsets $S_1,\ldots,S_k$ of $[n]$ with $\bigcup_{i=1}^k S_i = [n]$, such that $\mathcal{R}(S_1,\ldots,S_k,A)$ is minimized.
\end{definition}

\section{An FPTAS for \textit{k}-Subset Sum Ratio}
\label{sec:FPTAS_kSSR}

We define a restricted version of $\kSSRshort$, called $\kSSRRshort$, in which the largest element\footnote{Throughout the paper, the term \emph{element} of a set refers to the index $1\leq i\leq n$ instead of the value $a_i$, unless explicitly stated otherwise.} of the first set ($S_1$) is forced to be the smallest among the maxima of all $k$ sets and equal to a given number $p$. This definition generalizes the \textit{Semi-Restricted Subset-Sums Ratio} of \cite{SSR_MP}.

\begin{definition}[$\kSSRRshort$]
\label{def:kSSRR}
Given a sorted multiset $A=\{a_1,\ldots,a_n\}$ of positive integers and an integer $p,1 \leq p \leq n-k+1$, find $k$ disjoint subsets $S_1,\ldots,S_k$ of $[n]$ with $\max(S_1)=p$ and $\max(S_i)>p$ for $1<i \leq k$, such that $\mathcal{R}(S_1,\ldots,S_k,A)$ is minimized.\footnote{The case $p>n-k+1$ would result in the instance having no feasible solution.} 
\end{definition}

Most of this section is dedicated to producing an FPTAS for the restricted version $\kSSRRshort$, which is subsequently used as a subroutine to obtain an FPTAS for $\kSSRshort$, by iterating over all possible values of $p$. Similar reductions to a version with restricted largest elements have been used in approximation schemes for $\SSR$ \cite{SSR_Nan,SSR_MP,SSR_Bri}. In our paper, this technique is used to guarantee that each set $S_i$ contains a sufficiently large element, which is critical to ensure that the algorithm is an FPTAS. 

\subsection{Properties of optimal solutions to \texorpdfstring{\textit{k}-SSR\textsubscript{R}}{k-SSRR} instances}
\label{subsec:opt}

First, we want to guarantee that there is always an optimal solution to $\kSSRRshort$ that satisfies a few important properties. Recall that the input is sorted. Define $Q=\sum_{i=1}^p a_i$ and $q=\max\{i\mid a_i\leq Q\}$. Since $Q=\Sigma([p],A)$ and $S_1\subseteq [p]$, the following is immediate by Def.~\ref{def:max_min}.

\begin{observation}
\label{obs:smallest_set}
    For all feasible solutions $(S_1,\ldots,S_k)$ to a $\kSSRRshort$ instance $(A,p)$ $$m(S_1,\ldots,S_k,A) \leq \Sigma(S_1, A) \leq Q.$$
\end{observation}
Observations~\ref{obs:ratio_change} and~\ref{obs:smallest_set} are used in the proof of the following theorem to transform solutions without increasing their ratio.

\begin{theorem}
\label{thrm:optimal_sol}
    For any $\kSSRRshort$ instance $(A,p)$ there exists an optimal solution whose sets satisfy the following:
    \begin{enumerate}
        \item For all sets $S_i$ containing only elements $j$ with $a_j \leq Q$ it holds that $\Sigma(S_i, A) < 2Q$. \label{prop1}
        \item 
        Every set containing an element $j$ with $a_j > Q$ is a singleton. \label{prop2}
        \item         
        The union of singleton sets $\{j\}$ s.t. $a_j > Q$ is $\bigcup_{i=1}^x \{q+i\}$, where $x \geq 0$ is the number of these singleton sets. 
        \label{prop3}
    \end{enumerate}
\end{theorem}

\begin{proof}
Let $S=(S_1,\ldots,S_k)$ be an arbitrary feasible solution. If it violates property~\ref{prop1}, i.e. there exists a set $S_i$ in $S$ that only contains elements $j$ with $a_j \leq Q$ and has sum $\Sigma(S_i, A) \geq 2Q$, we transform it as follows. For all $j\in S_i$, we have
$$m(S,A)\leq \Sigma(S_1, A) \leq Q \leq \Sigma(S_i \setminus \{j\}, A) < \Sigma(S_i, A) \leq M(S,A).$$
Thus, if we remove an element $j$ from $S_i$, the ratio cannot increase. We remove the smallest element from $S_i$. As long as the solution still has a set $S_i$ violating property~\ref{prop1}, we can apply the same process repeatedly, until there are no more such sets. Since we are only removing elements, no new sets violating property~\ref{prop1} may appear during this process. Note that the largest element of every set remains intact, therefore the $\kSSRRshort$ restrictions $max(S_1) = p$ and $\max(S_i)>p$ for $1<i \leq k$ are satisfied.

If the new solution $S'$ violates property~\ref{prop2}, i.e. it contains a set $S_i$ with an element $j$ s.t. $a_j > Q$ and $|S_i|>1$, we apply the following transformation. It holds that
$$m(S',A)\leq \Sigma(S_1, A) \leq Q < \Sigma(\{j\}, A) < \Sigma(S_i, A) \leq M(S',A).$$
As such, replacing set $S_i$ with set $\{j\}$ will result in equal or smaller ratio. We repeat this for every such set $S_i$, thus yielding a solution in which every such set is a singleton. Note that the derived solution is a feasible $\kSSRRshort$ solution and it still satisfies property~\ref{prop1}, since we did not modify sets that contain only elements $j$ with $a_j \leq Q$.

Finally, if the new solution $S''$ violates property~\ref{prop3}, i.e. it contains a singleton set $\{u\}$ s.t. $a_u > Q$ and there is an unselected\footnote{An element that is not in any set $S_i$ of the solution.} element $v$ such that $q<v<u$, we do the following. It holds that
$$m(S'',A)\leq \Sigma(S_1, A) \leq Q < \Sigma(\{v\}, A) \leq \Sigma(\{u\}, A) \leq M(S'',A).$$
Thus, selecting $v$ instead of $u$ does not increase the ratio. Repeating this for all $u$ and $v$ satisfying the above mentioned property forces the union of these singleton sets to contain the smallest indices possible, yielding a feasible solution that satisfies all three properties.

In conclusion, for any feasible solution we can find another one that satisfies all properties~\ref{prop1},~\ref{prop2},~\ref{prop3} and has equal or smaller ratio. Thus, applying this to an arbitrary optimal solution proves the theorem.
\end{proof}

\subsection{A pseudo-polynomial time algorithm for \texorpdfstring{\textit{k}-SSR\textsubscript{R}}{k-SSRR}}
\label{subsec:DP}

In this subsection we present an exact algorithm for $\kSSRRshort$, which runs in $O(nQ^{k-1})$ time and returns an optimal solution satisfying the properties of Theorem~\ref{thrm:optimal_sol}. We present the algorithm in two parts for simplicity. Algorithm~\ref{alg:exact_k-SSRR} picks certain singleton sets with large elements and calls Algorithm~\ref{alg:DP_k-SSRR} to obtain candidate partial solutions for the remaining sets. We describe each algorithm in detail before providing its respective pseudocode.

Algorithm~\ref{alg:exact_k-SSRR} iterates over all possible values for the number $x$ of singleton sets $\{j\}$, with $a_j > Q$, specifically $x\in \{0,\ldots,\min \{k-1, n-q\} \}$. To construct these $x$ sets, it picks the smallest elements of $A$ that are greater than $Q$, in accordance to property~\ref{prop3} of Theorem~\ref{thrm:optimal_sol}. For each~$x$, Algorithm~\ref{alg:exact_k-SSRR} then calls Algorithm~\ref{alg:DP_k-SSRR} as a subroutine; the latter uses dynamic programming (DP) to output partial solutions for $A' = \{a_i \in A \mid a_i \leq Q\}$ and $k'=k-x$. Note that Algorithm~\ref{alg:DP_k-SSRR} is only called in cases where $p \leq |A'| - k' +1 = q - k+x +1$; otherwise, there is no feasible $\kSSRRshort$ solution for that value of $x$ (see Definition~\ref{def:kSSRR}).
 
 Algorithm~\ref{alg:exact_k-SSRR} appends the $x$ precalculated singletons to each partial solution returned by Algorithm~\ref{alg:DP_k-SSRR} and compares the ratio of each solution in order to find the best one. The $\kSSRRshort$ solution returned by Algorithm~\ref{alg:exact_k-SSRR} is optimal, as will be shown in Theorem~\ref{thrm:kSSRR_exact_correctness}.

\begin{algorithm}[H] 
\caption{$\texttt{Exact\textunderscore $\kSSRRshort$}(A,p)$}
\label{alg:exact_k-SSRR}
\begin{algorithmic}[1]

\Require{A sorted multiset $A=\{a_1,\ldots,a_n\}, a_i\in\mathbb{Z}^+$ and an integer $p, 1 \leq p \leq n-k+1$.}

\Ensure{Disjoint subsets $(S_1,\ldots,S_k)$ of $[n]$ with $\max(S_1)=p$, $\max(S_i)>p$ for $1<i \leq k$, minimizing $\mathcal{R}(S_1,\ldots,S_k,A)$.}

\State $best\_ratio \leftarrow \infty$, $best\_solution \leftarrow (\emptyset,\ldots,\emptyset)$

\State    $Q \leftarrow \sum_{i=1}^p a_i$,         
    $q \leftarrow \max\{i\mid a_i \leq Q\}$, 
    $x_{max} \leftarrow \min \{k-1, n-q\}$
    
    \For{$x \leftarrow 0, \ldots, x_{max}$}

        \State \textbf{for} $y\leftarrow 1, \ldots, x$ \textbf{do}
        $S_{k-x+y} \leftarrow \{q+y\}$\Comment{$x$ singletons}
        \EndFor

        \State $A' \leftarrow \{a_1,\ldots,a_q\}$, $k' \leftarrow k-x$
\algstore{Exact-SSRR}
\end{algorithmic}
\end{algorithm}

\begin{algorithm}[H]
\begin{algorithmic}[1]
\algrestore{Exact-SSRR}

        \If{$p \leq q-k'+1$}
            
            \State $DP\_solutions \leftarrow \texttt{DP\textunderscore $\kSSRRshort$}(A',k',p)$  \Comment{Call Alg.~\ref{alg:DP_k-SSRR} for partial solutions}        
        
        \For{\textnormal{\textbf{all}} $(S_1,\ldots,S_{k'})$ \textnormal{\textbf{in}} $DP\_solutions$}
            \State $current\_ratio \leftarrow \mathcal{R}(S_1,\ldots,S_k,A)$ \Comment{Ratio of all $k=k'+x$ sets}

            \If{$current\_ratio < best\_ratio$}
                \State $best\_solution \leftarrow (S_1,\ldots,S_k)$, 
                $best\_ratio \leftarrow current\_ratio$
            \EndIf
        \EndFor
        \EndIf
\State \Return $best\_solution$

\end{algorithmic}
\end{algorithm}

Algorithm~\ref{alg:DP_k-SSRR} draws from techniques of \cite{SSR_MP} and \cite{SSR_Nan}, while employing a more involved DP formulation to address the multi-subset setting. 
A DP table $T$ is used to systematically construct partial solution sets, 
while ensuring that the constraints of $\kSSRRshort$ are satisfied. 

Each cell of the DP table stores a tuple $(S_1,\dots,S_k,sum_1)$, where $S_j$ ($1 \leq j \leq k $) are the sets associated with said cell and $sum_1 = \Sigma(S_1,A)$. The coordinates of a cell $T_i[D][V]$ consist of the following components:

\begin{enumerate}
    \item An index $i$, indicating the number of elements examined so far.
    \item A \emph{difference vector} $D =  [d_2,d_3,\dots,d_k]$ (sorted in nondecreasing order) that encodes the differences between the sum of $S_1$ and the sums of the other sets, i.e. $d_j = \Sigma(S_1,A) - \Sigma(S_j,A)$. 
    \item A Boolean \emph{validity vector} $V =  [v_2,v_3,\dots,v_k]$, where $v_j$ is \texttt{true} iff $\max(S_j)>p$.
\end{enumerate}

The differences are used as coordinates in the DP table instead of the sums of the sets, in order to decrease the complexity of the algorithm by an order of $Q$ (see Lemma~\ref{lem:kSSRR_exact_complexity}). We also introduce $V$ as a coordinate, which is crucial to ensure optimality without increasing the complexity of the algorithm; essentially, the ratio of two tuples can be compared only if all of their validity Booleans are equal (see Lemma~\ref{lem:DP_conflicts}).

We now describe the DP process in a bottom-up manner. All cells $T_i[D][V]$ are initialized to empty tuples, except the cell $T_0[a_p,\ldots,a_p] [\texttt{false},\ldots,\texttt{false}]$, which is initialized to $(\{p\}, \emptyset, \ldots, \emptyset, a_p)$, thus forcing $p$ to be contained in $S_1$. Every element $i\in [q]$ is processed in increasing order. Assuming row $T_{i-1}$ contains tuples constructed by using elements up to $i-1$ (along with $p$, which is always in $S_1$), the next row $T_i$ is filled as follows. For each non-empty tuple $C=(S_1,\dots,S_k,sum_1)$ contained in some cell $T_{i-1}[D][V]$, $D=[d_2,\dots,d_k]$, $V=[v_2,\dots,v_k]$, we sequentially consider the cases below:
\begin{itemize}
    \item Element $i$ is not added to any set of $C$. In this case, $C$ is copied into $T_i[D][V]$.
    \item Element $i$ is added to $S_1$ (only if $i<p$). This produces a new tuple $C'=(S_1 \cup \{i\},\dots,S_k,sum_1+a_i)$ to be inserted into $T_i[D'][V]$, where $D'=[d_2',\dots,d_k']$ is the appropriately updated difference vector, i.e. $d_j'=d_j + a_i$, $1<j\leq k$.
    \item Element $i\neq p$ is added to each set $S_j$ ($1<j\leq k$), one at a time. This produces a new tuple $(S_1,\dots,S_j \cup \{i\},\ldots,S_k,sum_1)$, a difference vector $[d_2',\dots,d_k']$ and a validity vector $[v_2',\dots,v_k']$. Specifically, $v_j'$ is set to \texttt{true} if $i>p$ and $d_j'$ is set to $d_j - a_i$, while $v_u'=v_u$ and $d_u'=d_u$ for all $u\neq j$. This difference vector is then sorted in nondecreasing order, while the sets and the validity vector are rearranged accordingly, thus producing a tuple $C'$ and vectors $D',V'$. The tuple $C'$ is then inserted into $T_i[D'][V']$.
\end{itemize}

Algorithm~\ref{alg:DP_k-SSRR} prevents the new tuple from being inserted into a cell in the following cases:
\begin{enumerate}
    \item The cell already contains another tuple with equal or larger $sum_1$ (see Lemma~\ref{lem:DP_conflicts}).
    \item The updated vector $D'$ contains a difference smaller than or equal to $-2Q$. Such a tuple can be ignored, due to property~$\ref{prop1}$ of Theorem~\ref{thrm:optimal_sol}.
\end{enumerate}

In the end, Algorithm~\ref{alg:DP_k-SSRR} returns all non-empty tuples contained in cells $T_q[D][\textnormal{\texttt{true},\ldots,\texttt{true}}]$ (for all $D$), thus ensuring that every partial solution returned to Algorithm~\ref{alg:exact_k-SSRR} complies with the $\kSSRRshort$ restrictions. This concludes the description of Algorithm~\ref{alg:DP_k-SSRR}.

\begin{algorithm}[H] 
\caption{$\texttt{DP\textunderscore $\kSSRRshort$}(A,k,p)$}
\label{alg:DP_k-SSRR}
\begin{algorithmic}[1]

\Require{A sorted multiset $A=\{a_1,\ldots,a_q\}, a_i\in\mathbb{Z}^+$, an integer $k \geq 1$ and an integer $p,1 \leq p \leq q-k+1$, with the restriction $a_q \leq \sum_{i=1}^p a_i$.}

\Ensure{A set $solutions$ containing tuples of $k$ disjoint subsets $(S_1,\ldots,S_k)$ of $[q]$, with $\max(S_1)=p$, $\max(S_i)>p$ for $1<i \leq k$.}

\State $Q \leftarrow \sum_{i=1}^p a_i$

\State $T_i[D][V] \leftarrow \emptyset$ \textbf{for all} $1 \leq i \leq q$, $D=[d_2,\ldots,d_k]$, $V=[v_2,\ldots,v_k]$ with $d_2 \leq \ldots \leq d_k$

\State $T_0[a_p,\ldots,a_p] [\texttt{false},\ldots,\texttt{false}] \leftarrow (\{p\}, \emptyset, \ldots, \emptyset, a_p)$

\Comment{$S_1=\{p\}, S_j = \emptyset$, $1 < j \leq k$, $sum_1=a_p$}

\For{$i \leftarrow 1,\ldots,q$}

     \If{$i=p$}\Comment{$p$ is already in $S_1$, just copy $T_{i-1}$}
        \State $T_i[D][V] \leftarrow T_{i-1}[D][V]$ \textbf{for all} $T_{i-1}[D][V] \neq \emptyset$
        
        \State \textbf{continue} to next $i$
    \EndIf

    \For{\textnormal{\textbf{all}} $T_{i-1}[D][V] \neq \emptyset$}

        \If{$T_{i}[D][V] = \emptyset$ $\vee$ $T_{i-1}[D][V].sum_1>T_{i}[D][V].sum_1$}
        \State $T_i[D][V] \leftarrow T_{i-1}[D][V]$ \Comment{do not use element $i$}
        \EndIf

        \If{$i<p$} \Comment{add $i$ to $S_1$ only if $i<p$}
        \State $(S_1, \ldots, S_k, sum_1'), D', V' \leftarrow \texttt{update}(a_i,i,1,D,V,T_{i-1}[D][V])$
                
            \If{$T_{i}[D'][V'] = \emptyset$ $\vee$ $sum_1'>T_{i}[D'][V'].sum_1$}
                    \State $T_{i}[D'][V'] \leftarrow (S_1, \ldots, S_k, sum_1')$
            \EndIf
        \EndIf
        
        \For{$j \leftarrow 2, \ldots, k$} \Comment{add $i$ to $S_j$ $(j>1)$}
            \If{$d_j - a_i > -2Q$} \algorithmiccomment{prune if $d_j$ would become $\leq -2Q$}
            
                \State \hbox{$(S_1, \ldots, S_k, sum_1'), D', V' \leftarrow \texttt{update}(a_i,i,j,D,V,T_{i-1}[D][V])$}
                
                \If{$T_{i}[D'][V'] = \emptyset$ $\vee$ $sum_1'>T_{i}[D'][V'].sum_1$}
                \State $T_{i}[D'][V'] \leftarrow (S_1, \ldots, S_k, sum_1')$
                \EndIf
            \EndIf
        \EndFor    
    \EndFor    
\EndFor

\algstore{DP_k-SSRR}
\end{algorithmic}
\end{algorithm}

\begin{algorithm}[H]
\begin{algorithmic}[1]
\algrestore{DP_k-SSRR}
\State $solutions \leftarrow \emptyset$

\For{\textnormal{\textbf{all}} $T_q[D][\textnormal{\texttt{true},\ldots,\texttt{true}}] \neq \emptyset$}
    \Comment{consider only $T_q$ with $v_j = \texttt{true}$, $\forall j$}
    \State $(S_1, \ldots, S_k, sum_1) \leftarrow T_q[D][\texttt{true},\ldots,\texttt{true}]$
        
    \State $solutions \leftarrow solutions \cup \{(S_1, \ldots, S_k)\}$
\EndFor
\State \Return $solutions$
\end{algorithmic}
\end{algorithm}

Note that in cases where Algorithm~\ref{alg:DP_k-SSRR} is called with input $k=1$, the difference and validity vectors are empty, so each $T_i$ consists of just one DP cell.

It remains to define the \texttt{update} function used by Algorithm~\ref{alg:DP_k-SSRR}. This function is tasked with adding element $i$ to $S_j$ and updating all fields affected by this addition. Specifically, $D$ is modified by adding $a_i$ to all differences (if $j=1$) or subtracting $a_i$ from $d_j$ (if $j \neq 1)$. Next, $sum_1$ is updated (if $j=1$) and $v_j$ is set to \texttt{true} if $i>p$ and $j>1$. Lastly, $D$ is sorted in ascending order, while the sets and validity values associated with each difference are rearranged accordingly. This function (Algorithm~\ref{alg:update}) runs in time $O(k)$, since (at most) one $d_j$ needs to be rearranged. 

\begin{algorithm}[H] 
\caption{\texttt{update}$(a_i,i,j,D,V,S_1,\ldots,S_k,sum_1)$}
\label{alg:update}
\begin{algorithmic}[1]

\Require{An element $a_i$ and its index $i$, an index $j$ ($1\leq j \leq k$), a sorted difference vector $D=[d_2,\ldots,d_k]$, a Boolean validity vector $V=[v_2,\ldots,v_k]$, $k$ disjoint sets $S_j$ containing elements in $[i-1]$ and the sum $sum_1$ of $S_1$.}

\Ensure{$k$ updated sets $S_j$, the updated sum $sum_1$ of $S_1$, the updated (sorted) difference vector $D$ and the updated validity vector $V$.}



\If{$j = 1$}
    \State $S_1 \leftarrow S_1 \cup \{i\}, \; sum_1 \leftarrow sum_1 + a_i$

    
    \For{\textbf{all} $d_j$ in $D$}
        \State $d_j \leftarrow d_j+a_i$
    \EndFor
    
    \State \Return $(S_1,\dots,S_k,sum_1),\;  D,\;  V$
\EndIf

\State $S_j \leftarrow S_j \cup \{ i\}, \; d_j \leftarrow d_j-a_i$

\If{$i > p \;$}
    \State $v_j\leftarrow\texttt{true}$
\EndIf

\While{$d_{j-1} > d_j$}
\Comment{Swap the sets and their respective $d$, $v$}
\State \texttt{swap}$(S_j,S_{j-1})$,
\texttt{swap}$(d_j,d_{j-1})$, \texttt{swap}$(v_j,v_{j-1})$
    


    \State $j \leftarrow j-1$
    
\EndWhile

\State \Return $(S_1,\dots, S_k,sum_1), \; D, \; V$

\end{algorithmic}
\end{algorithm}



We say that a \emph{conflict} occurs when a tuple is to be stored in a cell $T_i[D][V]$ which is already occupied by another (non-empty) tuple. Note that conflicting tuples must have the same difference and validity vectors and only use elements up to $i$, by definition. Since the DP process must minimize the overall ratio, including potential singletons not participating in said process, conflict resolution becomes more involved for $k>2$.

\begin{lemma}[Conflict resolution]
\label{lem:DP_conflicts}  
    Let $C_1 = (S_1,\ldots,S_k,sum_1)$ and $C_2 = (S_1',\ldots,S_k',sum_1')$ be two tuples that result in a conflict in a cell $T_i[D][V]$ of the DP table of Algorithm $\ref{alg:DP_k-SSRR}$, such that $sum_1 < sum_1'$. No optimal $\kSSRRshort$ solution may use\footnote{In our dynamic programming framework, we say that a solution $S$ \emph{uses} a tuple $C=(S_1,\ldots,S_k,sum_1)$ if $C$ appears in an intermediate step of the construction of $S$.} $C_1$. 
\end{lemma}

\begin{proof}
    Let $S$ be a $\kSSRRshort$ solution. Split $S$ in two: $S_{small}$, containing the $k'$ sets returned by Algorithm~\ref{alg:DP_k-SSRR}, and $S_{large}$, containing the singletons $\{j\}$ s.t. $a_j > Q$. Assume $S_{small}$ uses~$C_1$.
    
    By assumption, $C_1$ and $C_2$ have identical $D$ and $V$ vectors and their sets involve only elements from $[i]$. This implies that any combination of elements larger than $i$ that can later be added to sets of $C_1$ to construct a partial\footnote{We denote as \emph{partial} a solution for which $S_{large}$ will be appended to obtain a feasible $\kSSRRshort$ solution.} solution with vectors $D',V'$ can also be added to the corresponding sets of $C_2$ to construct a partial solution with the same vectors $D',V'$.
    Let $S_{small}'$ be the partial solution obtained by using $C_2$ instead of $C_1$ and adding the same combination of elements that would have been used to obtain $S_{small}$ from $C_1$. 
    Consider the solution $S'$ that is obtained by appending $S_{large}$ to $S_{small}'$. Note that $S_{small}$ and $S_{small}'$ have the same validity vector, hence $S'$ is also a feasible $\kSSRRshort$ solution.

    Let $M_s,m_s$ be the maximum and minimum sum of $S_{small}$ respectively and $M_s',m_s'$ those of $S_{small}'$. Since $sum_1 < sum_1'$ and the difference vector $D$ is the same for both solutions, it follows that $M_s<M_s'$, $m_s<m_s'$ and $M_s-m_s=M_s'-m_s'$. From these, we obtain
\begin{equation}
\label{ineq:ratio_comparison}
    \frac{M_s}{m_s}>\frac{M_s'}{m_s'}.
\end{equation}

    Let $\{t\} \in S_{large}$ be the singleton set with the largest element. For all $\{i\}\in S_{large}$ such that $i\neq t$, we have 
    $$m_s\leq \Sigma(S_1,A) \leq Q < \Sigma(\{i\},A) \leq \Sigma(\{t\},A).$$ 
    This implies that $\{i\}$ ($i\neq t$) does not affect the ratio and $\{t\}$ only affects the ratio if $a_t > M_s$. Thus, all the sets in $S_{large}$ can be ignored when calculating $\mathcal{R}(S,A)$, except for $\{t\}$ if $a_t > M_s$. The same holds for $\mathcal{R}(S',A)$ and $M_s'$. Consider three cases:

    \begin{enumerate}
        \item $M_s<M_s'<a_t$. In this case, $\mathcal{R}(S,A)=a_t/m_s$ and $\mathcal{R}(S',A)=a_t/m_s'$. Since $m_s<m_s'$: $\mathcal{R}(S,A)>\mathcal{R}(S',A)$.
        \item $M_s<a_t \leq M_s'$. In this case, $\mathcal{R}(S,A)=a_t/m_s>M_s/m_s$ and $\mathcal{R}(S',A)=M_s'/m_s'$. By inequality \eqref{ineq:ratio_comparison}: $\mathcal{R}(S,A)>\mathcal{R}(S',A)$.
        \item $a_t \leq M_s<M_s'$. In this case, $\mathcal{R}(S,A)=M_s/m_s$ and $\mathcal{R}(S',A)=M_s'/m_s'$. By inequality~\eqref{ineq:ratio_comparison}: $\mathcal{R}(S,A)>\mathcal{R}(S',A)$.
    \end{enumerate}
    
    In all cases $\mathcal{R}(S,A)>\mathcal{R}(S',A)$, therefore $S$ cannot be optimal.
\end{proof}

\begin{lemma}[Feasibility]
\label{lem:kSSRR_feasibility}
    Every $k$-tuple of sets considered by Algorithm $\ref{alg:exact_k-SSRR}$ is a feasible $\kSSRRshort$ solution. 
\end{lemma}
\begin{proof}
The $x$ singletons constructed by Algorithm~\ref{alg:exact_k-SSRR} are disjoint and contain elements larger than $q$. In Algorithm~\ref{alg:DP_k-SSRR}, every element \(i\leq q\) is processed in increasing order and added to at most one set at a time. As such, all sets are disjoint. Note that \(p\) is contained in \(S_1\) and $S_1$ cannot receive a larger element.  Recall that $v_i$ is set to \texttt{true} when $S_i$ receives an element $j>p$ and Algorithm~\ref{alg:DP_k-SSRR} only returns solutions with \(V=[\texttt{true},\dots,\texttt{true}]\). Thus, all $\kSSRRshort$ restrictions regarding the largest element of each set are satisfied.
\end{proof}

The proof of the following theorem uses Lemmas~\ref{lem:DP_conflicts} and~\ref{lem:kSSRR_feasibility}. The main idea is to show that Algorithm~\ref{alg:exact_k-SSRR} considers and therefore finds the optimal solution whose existence is guaranteed by Theorem~\ref{thrm:optimal_sol}. 

\begin{theorem}[Optimality]
\label{thrm:kSSRR_exact_correctness}
    Algorithm $\ref{alg:exact_k-SSRR}$ returns an optimal solution for $\kSSRRshort$.
\end{theorem}

\begin{proof}

Lemma~\ref{lem:kSSRR_feasibility} guarantees that Algorithm~\ref{alg:exact_k-SSRR} returns a feasible $\kSSRRshort$ solution, so we only need to prove that its ratio is optimal. Let $S^*$ be the optimal $\kSSRRshort$ solution guaranteed by Theorem~\ref{thrm:optimal_sol}. We split $S^*$ in two: $S_{large}^*$, which contains the singleton sets $\{j\}$ s.t. $a_j > Q$, and $S_{small}^*$, which contains the rest of the sets. Note that the number $x$ of these singletons is bounded by $k-1$ or by the amount of sufficiently large elements, i.e. $n-q$. Thus, Algorithm~\ref{alg:exact_k-SSRR} considers every $S_{large}$ that satisfies properties~\ref{prop2} and~\ref{prop3} of Theorem~\ref{thrm:optimal_sol}, including $S_{large}^*$.

Recall that Algorithm~\ref{alg:exact_k-SSRR} uses Algorithm~\ref{alg:DP_k-SSRR} as a subroutine in order to obtain candidate partial solutions, to which $S_{large}^*$ will be appended. We would like $S_{small}^*$ to be contained in the solutions returned by Algorithm~\ref{alg:DP_k-SSRR}.

Algorithm~\ref{alg:DP_k-SSRR} prunes a solution if a difference $d_j$ would become lower than or equal to $-2Q$. Any $S_{small}$ satisfying property~\ref{prop1} of Theorem~\ref{thrm:optimal_sol} will not be affected by this. This includes $S_{small}^*$.

When there is a conflict in a cell $T_i[D][V]$ between two tuples with different $sum_1$ values, $S_{small}^*$ cannot use the one with smaller $sum_1$ (by Lemma $\ref{lem:DP_conflicts}$), since $S^*$ is a optimal. When there is a conflict in a cell $T_i[D][V]$ between two tuples $C_1,C_2$ with equal $sum_1$ values, observe that both tuples have exactly the same sums (but different sets), use only elements up to $i$ and have the same validity vector. This implies that any combination of elements greater than $i$ that can later be added to sets of $C_1$ can also be added to the same sets of $C_2$ and vice versa. Thus, any combination of sums that can be reached through $C_1$ can also be reached through $C_2$ and vice versa. This means that if one of these tuples leads to $S_{small}^*$, the other one leads to another partial solution $S_{small}^{**}$, which has exactly the same sums as $S_{small}^*$. Appending $S_{large}^*$ to $S_{small}^{**}$ yields a solution with the same sums as $S^*$, i.e. another optimal solution.

It follows directly from the description of Algorithm~\ref{alg:DP_k-SSRR} that it constructs every possible combination of disjoint sets $S_1,\ldots,S_{k-x}$ with $\max (S_1) = p$, apart from the ones pruned as explained in the previous paragraphs.

Therefore, the partial solutions returned by Algorithm~\ref{alg:DP_k-SSRR} contain $S_{small}^*$ or another partial solution with the same sums, which also leads to an optimal solution when appended with $S_{large}^*$. In either case, Algorithm~\ref{alg:exact_k-SSRR} will find an optimal solution when iterating to find the best ratio.
\end{proof}

\begin{lemma}[Complexity]
\label{lem:kSSRR_exact_complexity}
    Algorithm $\ref{alg:exact_k-SSRR}$ runs in time $O(nQ^{k-1})$.
\end{lemma}
\begin{proof}  
    Algorithm~\ref{alg:exact_k-SSRR} calls Algorithm~\ref{alg:DP_k-SSRR} once for each value of $x$ in $\{0,\ldots,\min\{k-1, n-q\}\}$, where $x$ is the number of singleton sets $\{j\}$ s.t. $a_j > Q$. For a fixed $x$, Algorithm~\ref{alg:DP_k-SSRR} is applied to an instance with $k' = k - x$ sets and \(q=O(n)\) elements.

    Due to the pruning done in Algorithm~\ref{alg:DP_k-SSRR}, it holds that $\forall j\in\{2,\ldots,k'\}$: $-2Q < d_j\leq Q$, where the second inequality holds because $\Sigma(S_1,A) \leq Q$. Since the algorithm stores \(D\) in non-decreasing order (i.e., \(d_2 \le d_3 \le \cdots \le d_{k'}\)), there are $O((3Q)^{k'-1}/(k'-1)!)$ distinct difference vectors. Taking into account the validity vector $V$ and all $T_i$'s, we obtain the following bound for the amount of DP cells: $O(n(6Q)^{k'-1}/(k'-1)!)$.

    For constant $k$, this bound becomes $O(nQ^{k'-1})$. Note that all operations of Algorithm~\ref{alg:DP_k-SSRR} on DP cells (such as sorting the difference vector after adding an element) run in time $O(k')$, so the time complexity of Algorithm~\ref{alg:DP_k-SSRR} is $O(nQ^{k'-1})=O(nQ^{k-x-1})$. Hence, the time complexity of Algorithm~\ref{alg:exact_k-SSRR} is bounded by:
    $$\sum_{x=0}^{k-1} nQ^{k-x-1}=O(nQ^{k-1})$$
\end{proof}

\subsection{FPTASs for \texorpdfstring{$\kSSRRshort$}{k-SSRR} and \texorpdfstring{$\kSSRshort$}{k-SSR}}
\label{subsec:kSSRR_FPTAS}

By Theorem~\ref{thrm:kSSRR_exact_correctness} and Lemma~\ref{lem:kSSRR_exact_complexity}, Algorithm~\ref{alg:exact_k-SSRR} solves $\kSSRRshort$ in pseudo-polynomial time. To obtain an FPTAS for $\kSSRRshort$, we scale (and round down) the input set by a factor $\delta = \frac{\varepsilon \cdot a_p}{3\cdot n}$~(cf. \cite{SSR_MP,SSR_Nan}).

This leads to Algorithm \ref{alg:FPTAS_k-SSRR}, which calls Algorithm~\ref{alg:exact_k-SSRR} in order to find an optimal solution for the rounded input $A_r$, yielding a $(1+\varepsilon)$-approximation of the (largest-to-smallest) ratio of an optimal solution $(S^*_1,\ldots,S^*_k)$ of the original $\kSSRRshort$ instance.

\begin{algorithm}[H] 
\caption{$\texttt{FPTAS\textunderscore $\kSSRRshort$}(A,p,\varepsilon)$}
\label{alg:FPTAS_k-SSRR}
\begin{algorithmic}[1]

\Require{A sorted multiset $A=\{a_1,\ldots,a_n\}, a_i\in\mathbb{Z}^+$, an integer $p$ such that $1 \leq p \leq n-k+1$ and an error parameter $\varepsilon \in \left(0, 1\right)$.}

\Ensure{Disjoint subsets $(S_1,\ldots,S_k)$ of $[n]$ with $\max(S_1)=p$, $\max(S_i)>p$ for $1<i \leq k$ and $\mathcal{R}(S_1,\dots,S_k,A) \leq (1+\varepsilon)\cdot \mathcal{R}(S^*_1,\ldots,S^*_k,A)$.}

\State $\delta \leftarrow \frac{\varepsilon \cdot a_p}{3\cdot n}$, $A_r\leftarrow \emptyset$

\State \textbf{for} $i \leftarrow 1,\ldots,n$ \textbf{do}
     $a^r_i \leftarrow \lfloor \frac{a_i}{\delta} \rfloor$,   
    $A_r\leftarrow A_r\cup \{ a^r_i \}$
\State $(S_1,\dots,S_k) \leftarrow$ \texttt{Exact\textunderscore $\kSSRRshort$}$(A_r,p)$ \Comment{Call Alg. \ref{alg:exact_k-SSRR} for rounded instance}

\State \Return $S_1, \dots S_k$
\end{algorithmic}
\end{algorithm}

\begin{theorem}[$\kSSRRshort$ approximation]
\label{thrm:approximation_ratio}
Let $(S_1,\dots,S_k)$ be the sets returned by Algorithm~$\ref{alg:FPTAS_k-SSRR}$ for a $\kSSRRshort$ instance $(A=\{a_1,\ldots,a_n\}$, $p)$ with error parameter $\varepsilon$. Let $(S^*_1,\dots,S^*_k)$ be an optimal solution for the same $\kSSRRshort$ instance. Then:
$$\mathcal{R}(S_1, \dots, S_k,A) 
\leq (1+\varepsilon)\cdot \mathcal{R}(S^*_1, \ldots, S^*_k,A)$$
\end{theorem}

The proof of Theorem~\ref{thrm:approximation_ratio} is analogous to the respective proofs for $k=2$~\cite{SSR_Nan,SSR_MP}, with modifications to work for more than two subsets. Due to its length, we defer it to Subsection~\ref{subsec:proof_new}.

By Lemma~\ref{lem:kSSRR_exact_complexity}, Algorithm~\ref{alg:FPTAS_k-SSRR} solves $\kSSRRshort$ in $O(n Q^{k-1})$ 
(where $Q= \sum_{i=1}^p a^r_i$). Since we scaled the input by $\delta$, the values of Q are bounded as follows:
$$Q^{k-1}=\left(\sum_{i=1}^p a^r_i \right)^{k-1} \leq \left(n\cdot a^r_p\right)^{k-1} \leq \left(\frac{n\cdot a_p}{\delta} \right)^{k-1}= \left(\frac{3\cdot n^2}{\varepsilon}\right)^{k-1} = \frac{3^{k-1}\cdot n^{2k-2}}{\varepsilon^{k-1}}$$
Therefore, Algorithm~\ref{alg:FPTAS_k-SSRR} runs in time $O({n^{2k-1}}/{\varepsilon^{k-1}})$. To obtain an FPTAS for the (unrestricted) $\kSSRshort$ problem, we run Algorithm~\ref{alg:FPTAS_k-SSRR} once for each possible value of $p$ ($1 \leq p \leq n-k+1$) and pick the solution with the best ratio. Thus, we obtain the main theorem of this section.

\begin{theorem}
\label{thrm:main}
There is an FPTAS for $\kSSRshort$ that runs in $O({n^{2k}}/{\varepsilon^{k-1}})$ time. 
\end{theorem}

\subsection{Proof of Theorem~\ref{thrm:approximation_ratio}}
\label{subsec:proof_new}

Let $(A=\{a_1,\dots,a_n\},p)$ be a $\kSSRRshort$ instance and $\varepsilon \in (0,1)$ be an error parameter. We use $\delta = \frac{\varepsilon \cdot a_p}{3\cdot n}$ as a scale factor.

Let $A_r = \{a^r_i =\lfloor 
\frac{a_i}{\delta} \rfloor : for\; all\;  i \in \{1,\ldots,n\} \}$. Let $S_{alg}=(S_1,\dots,S_k)$ be the solution returned by Algorithm~\ref{alg:FPTAS_k-SSRR} and $S_{opt}=(S^*_1,\dots,S^*_k)$ be an optimal solution for the instance $(A,p)$. We define
$$S_M = \argmax_{S_i \in S_{alg}} \Sigma(S_i,A),
\quad
S_m =  \argmin_{S_i \in S_{alg}}\Sigma(S_i,A),$$
$$S_M^* = \argmax_{S^*_i \in S_{opt}} \Sigma(S^*_i,A),
\quad
S_m^* = \argmin_{S^*_i \in S_{opt}} \Sigma(S^*_i,A).$$
We continue by proving some auxiliary lemmas.

\begin{lemma}
\label{lem:ineq_approx_1}
For any $S \in \{ S_M,S_m, S_M^*,S_m^*\}$ it holds that
  $$\sum_{i \in S} a_i - n\cdot \delta \leq \sum_{i \in S} \delta \cdot a^r_i \leq \sum_{i \in S} a_i.$$
\end{lemma}

\begin{proof}
Recall that $a^r_i =\lfloor 
\frac{a_i}{\delta} \rfloor$ for all  
$i \in \{1,\ldots,n\}$. Hence

 $$\frac{a_i}{\delta} -1 \leq a^r_i \leq  \frac{a_i}{\delta}  \Rightarrow 
 a_i -\delta \leq \delta \cdot a^r_i \leq a_i 
 \Rightarrow
 \sum_{i \in S} (a_i -\delta) \leq \sum_{i \in S} \delta \cdot a^r_i \leq \sum_{i \in S} a_i.$$

Since $S \in \{ S_M,S_m,S_M^*,S_m^*\}$, we have $|S|\leq n$, thus proving the lemma. 
\end{proof}

\begin{lemma}
\label{lem:ineq_approx_2}
For any $S \in \{ S_M,S_m, S_M^*,S_m^*\}$ it holds that    
    $$n\cdot \delta \leq \frac{\varepsilon}{3} \cdot \sum_{i \in S} a_i.$$
\end{lemma}

\begin{proof}
    Note that $\max  S \geq p$ for any $S \in \{ S_M,S_m,S_M^*,S_m^*\}$, by definition of $\kSSRRshort$. Since the input is sorted, it holds that

    $$ n\cdot \delta = \frac{\varepsilon \cdot a_p}{3} \leq  \frac{\varepsilon}{3} \cdot \sum_{i \in S }a_i.$$
\end{proof}

\begin{lemma}
\label{lem:Salg_inequality}
It holds that $\mathcal{R}(S_{alg},A) \leq  
\mathcal{R}(S_{alg},A_r)+ \varepsilon/3$.
\end{lemma}

\begin{proof}
We can prove that
\begin{align*}
\mathcal{R}(S_{alg},A)  = \frac{\sum_{i \in S_M} a_i}{\sum_{j \in S_m} a_j}
				   & \leq \frac{\sum_{i \in S_M}\delta\cdot a^r_i + \delta \cdot n}{\sum_{j \in S_m}a_j} 
				   & \mbox{[by Lemma~}~\ref{lem:ineq_approx_1}\mbox{]} \\
				   & \leq \frac{\sum_{i \in S_M} a^r_i }{\sum_{j \in S_m} a^r_j} +  
				   \frac{ \delta \cdot n}{\sum_{j \in S_m} a_j}.
				   & \mbox{[by Lemma~}~\ref{lem:ineq_approx_1}\mbox{]} 
\end{align*}
By Lemma~\ref{lem:ineq_approx_2} it holds that $\delta \cdot n \leq (\varepsilon/3)\cdot \sum_{j \in S_m} a_j$. Thus
\begin{align*}
    &\mathcal{R}(S_{alg},A) \leq \frac{\sum_{i \in S_M} a^r_i }{\sum_{j \in S_m} a^r_j}+ \frac{\varepsilon}{3} = \mathcal{R}(S_{alg},A_r) + \frac{\varepsilon}{3}. & \mbox{[by Def.\;}~\ref{def:k-ratio}\mbox{]}
\end{align*}
\end{proof}

\begin{lemma}
\label{lem:Sopt_inequality}
For any $\varepsilon\in (0,1)$, it holds that
$$ \mathcal{R}(S_{opt},A_r) \leq \left(1+\frac{\varepsilon}{2}\right)\cdot \mathcal{R}(S_{opt},A).$$
\end{lemma}

\begin{proof}
We can prove that
\begin{align*}
 \mathcal{R}(S_{opt}&,A_r)   = \frac{\sum_{i \in S_M^*} a^r_i}{\sum_{j \in S_m^*} a^r_j}&\\
 & \leq \frac{\sum_{i \in S_M^*}a_i}{\sum_{j \in S_m^*}a_j - n\cdot \delta} &\mbox{[by Lemma~}\ref{lem:ineq_approx_1}\mbox{]}\\
 & = \frac{\sum_{j \in S_m^*} a_j}{\sum_{j \in S_m^*} a_j - n\cdot \delta} \cdot \frac{\sum_{i \in S_M^*} a_i}{\sum_{j \in S_m^*} a_j}& \\
 & = \left(1 + \frac{n\cdot \delta}{\sum_{j \in S_m^*} a_j - n\cdot \delta} \right) \cdot \frac{\sum_{i \in S_M^*}a_i}{\sum_{j \in S_m^*}a_j}.& \\
 \end{align*}
After using Lemma~\ref{lem:ineq_approx_2} for $S= S_m^*$ and simplifying the fraction, we obtain
 $$\mathcal{R}(S_{opt},A_r) \leq \left(1 + \frac{\varepsilon}{ 3- \varepsilon}\right)\cdot \frac{\sum_{i \in S_M^*}a_i}{\sum_{j \in S_m^*}a_j}.$$
Since $\varepsilon \in (0,1)$, we have
 \begin{align*}
  \mathcal{R}(S_{opt},A_r)&\leq \left(1 + \frac{\varepsilon}{ 2}\right)\cdot \frac{\sum_{i \in S_M^*}a_i} {\sum_{j \in S_m^*}a_j} &\\
 & = \left(1 + \frac{\varepsilon}{ 2}\right)\cdot \mathcal{R}(S_{opt},A). &\mbox{[by Def.\;}~\ref{def:k-ratio}\mbox{]}
\end{align*}
\end{proof}
We are now ready to prove Theorem~\ref{thrm:approximation_ratio}.

\setcounter{theorem}{13}
\begin{theorem}[$\kSSRRshort$ approximation]
Let $(S_1,\dots,S_k)$ be the sets returned by Algorithm~\ref{alg:FPTAS_k-SSRR} for a $\kSSRRshort$ instance $(A=\{a_1,\ldots,a_n\}$, $p)$ with error parameter $\varepsilon$. Let $(S^*_1,\dots,S^*_k)$ be an optimal solution for the same $\kSSRRshort$ instance. Then
$$ \mathcal{R}(S_1, \dots, S_k,A) 
\leq (1+\varepsilon)\cdot \mathcal{R}(S^*_1, \ldots, S^*_k,A). $$
\end{theorem}

\begin{proof}
The proof is a direct consequence of Lemmas~\ref{lem:Salg_inequality} and~\ref{lem:Sopt_inequality}. We use the notations $S_{alg}=(S_1,\dots,S_k)$ and $S_{opt}=(S^*_1,\dots,S^*_k)$ for simplicity.  We also use the fact that $S_{alg}$ is an optimal solution for $(A_r,p)$, according to Theorem~\ref{thrm:kSSRR_exact_correctness}.

\begin{align*}
 \mathcal{R}(S_{alg},A) 
 &\leq \mathcal{R}(S_{alg},A_r)+ \frac{\varepsilon}{3} \leq \mathcal{R}(S_{opt},A_r)+ \frac{\varepsilon}{3}\\
  &\leq \left(1+\frac{\varepsilon}{2}\right)\cdot \mathcal{R}(S_{opt},A) + \frac{\varepsilon}{3}\\
  &\leq \left(1+\varepsilon \right)\cdot \mathcal{R}(S_{opt},A)
\end{align*}
\end{proof}

\setcounter{theorem}{19}

\section{An FPTAS for k-way Number Partitioning Ratio}
\label{sec:FPTAS_kPR}

In this section we present an FPTAS for $\kPR$ ($\kPRshort$) by extending and refining the techniques of Section $\ref{sec:FPTAS_kSSR}$. Recall that in the case of $\kPRshort$, every element needs to be assigned to a set. We define the following restricted version of $\kPRshort$, which is analogous to $\kSSRRshort$.

\begin{definition}[$\kPRRshort$]
\label{def:kPRR}
Given a sorted multiset $A=\{a_1,\ldots,a_n\}$ of positive integers and an integer $p,1 \leq p \leq n-k+1$, find $k$ disjoint subsets $S_1,\ldots,S_k$ of $[n]$ with $\bigcup_{i=1}^k S_i = [n]$, $\max(S_1)=p$ and $\max(S_i)>p$ for $1<i \leq k$, such that $\mathcal{R}(S_1,\ldots,S_k,A)$ is minimized.
\end{definition}

The main challenge in expanding our technique to $\kPRshort$ stems from the fact that there exist instances $(A,p)$ of $\kPRRshort$ where all optimal solutions violate one or more of the properties of Theorem $\ref{thrm:optimal_sol}$. This is a problem, since the time complexity of our $\kSSRRshort$ algorithm relies heavily on Theorem $\ref{thrm:optimal_sol}$.

We overcome this by showing that there exists a  $p^*$ such that the corresponding optimal solution $S^*$ to $(A,p^*)$ is well-behaved (see Theorem~\ref{thrm:perfect_p}). It suffices to guarantee that for $p^*$, the algorithm will consider $S^*$ (or another equivalent solution) as a candidate solution for the rounded $\kPRRshort$ instance. Note that since $S^*$ is not guaranteed to be optimal for the rounded $\kPRRshort$ instance, we obtain \emph{neither an exact algorithm nor an FPTAS} for the restricted problem $\kPRRshort$.

\subsection{Properties of \texorpdfstring{\textit{k}-PART\textsubscript{R}}{k-PARTR} instances} 
\label{subsec:optP}

\begin{definition}[Perfect $p$]
\label{def:optimal_p}
    Let $A=\{a_1,\ldots,a_n\}$ be the input to the $\kPRshort$ problem. We call a number $p:1\leq p \leq n-k+1$ \emph{perfect} for $A$ if the optimal solution(s) for the $\kPRRshort$ instance $(A,p)$ have the same (largest-to-smallest) ratio as the optimal solution(s) for the $\kPRshort$ instance $A$.
\end{definition}

We define $Q=\sum_{i=1}^p a_i$ and $q=\max\{i\mid a_i\leq Q\}$, in the same manner as we did for $\kSSRRshort$. The next theorem is analogous to Theorem~\ref{thrm:optimal_sol}, with a key difference: its properties are only guaranteed for some \emph{perfect} $p$. Its proof uses arguments similar in spirit to the proof of Theorem~\ref{thrm:optimal_sol}. Special attention is required since elements that violate a property are reassigned to other sets rather than being removed entirely from the solution. This may cause some $\kPRRshort$ restriction to be violated, thus rendering the proof significantly more involved.

\begin{theorem}
\label{thrm:perfect_p}
    Given a $\kPRshort$ instance $A$, there exists a perfect $p$ for $A$, for which there is an optimal solution for the $\kPRRshort$ instance $(A,p)$ satisfying the following properties:
    \begin{enumerate}
        \item For all sets $S_i$ containing only elements $j \leq q$ it holds that $\Sigma(S_i, A) < 2Q$.\label{part_prop1}
        \item All elements $j>q$ are contained in singleton sets. \label{part_prop2}
    \end{enumerate}
\end{theorem}

\begin{proof}    
For some arbitrary value of $p$, let $S=(S_1,\ldots,S_k)$ be an arbitrary feasible solution for the $\kPRRshort$ instance $(A,p)$. If $S$ breaks property~\ref{part_prop2}, i.e. there exists a set $S_i$ in $S$ with $|S_i|>1$ containing an element $j > q$, we do the following. Let $S_m$ be a set with minimum sum\footnote{There could be multiple sets with the same (minimum) sum.} and $u$ the smallest element of $S_i$. The following inequalities hold.
    $$m(S,A) = \Sigma(S_m,A) \leq Q < \Sigma(S_i \setminus \{u\},A)$$
    $$\Sigma(S_m \cup \{u\},A) \leq Q+a_u < \Sigma(S_i,A) \leq M(S,A)$$
From these we can infer that, by moving $u$ from $S_i$ to $S_m$, the minimum sum cannot decrease and the maximum cannot increase. Thus, we move $u$ as described, yielding a solution $S'$ with $\mathcal{R}(S',A) \leq \mathcal{R}(S,A)$. However, it might be the case that $S_m$ is $S_1$ and adding $u$ to it breaks the restrictions of $\kPRRshort$, so $S'$ might not be a feasible solution for $(A,p)$. For the new solution $S'$, define $p'$ as the minimum among the maxima of its $k$ sets and call $S_1$ the set that contains $p'$. Note that $S'$ is a feasible solution for the $\kPRRshort$ instance $(A,p')$.

If $S'$ still breaks property~\ref{part_prop2} for $Q'=\sum_{i=1}^{p'} a_i$ and $q'=\max\{i\mid a_i\leq Q'\}$, apply the same process. By doing this repeatedly, $p$ cannot decrease and the ratio of the solution cannot increase. Note that $p$ cannot exceed $n-k+1$ with this process, so at some point $p$ will stop increasing. Let $p_f$ be this final value and define $Q_f,q_f$ accordingly for $p_f$. Repeating the process will at some point yield a solution that satisfies property~\ref{part_prop2} for the $\kPRRshort$ instance $(A,p_f)$.

Let $S''$ be the feasible solution for the $\kPRRshort$ instance $(A,p_f)$, derived from the process described in the previous paragraphs. $S''$ satisfies property~\ref{part_prop2} for values $Q_f,q_f$. If $S''$ breaks property~\ref{part_prop1}, i.e. it contains a set $S_i$ consisting only of elements $j\leq q_f$ and having sum $\Sigma(S_i,A)\geq 2Q_f$, we apply the following transformation. Let $S_m$ be a set with minimum sum. Take the smallest element $j$ from $S_i$ and move it to $S_m$. Observe that $a_j \leq Q_f$. Hence, the following inequalities hold.
    $$m(S'',A)=\Sigma(S_m,A) \leq Q_f \leq \Sigma(S_i \setminus \{j\},A)$$
    $$\Sigma(S_m \cup \{j\},A) \leq 2Q_f \leq \Sigma(S_i,A)\leq M(S'',A)$$
From these we can infer that, by moving $j$ from $S_i$ to $S_m$, the minimum sum cannot decrease and the maximum cannot increase. We repeatedly apply this process until we end up with a solution satisfying property~\ref{part_prop1}, just like we did in the previous paragraphs for property~\ref{part_prop2}, since the same arguments hold regarding the increment of $p$ and the ratio not increasing. Furthermore, this transformation preserves property~\ref{part_prop2}, as $p$ cannot decrease through this process and we only add elements to the set with the smallest sum (which cannot be a set containing an element $j>q$).

In conclusion, for any feasible solution for a $\kPRRshort$ instance $(A,p)$ that violates some property, we can find a feasible solution for another $\kPRRshort$ instance $(A,p_{new})$ that satisfies both properties and has equal or smaller ratio. Suppose we apply this to an optimal solution for a $\kPRRshort$ instance $(A,p)$, with $p$ being \emph{perfect}\footnote{Such a $p$ always exists, by definition.} for the $\kPRshort$ instance $A$. Since the ratio of the solution does not increase throughout the constructions described in this proof, any new $p'$ obtained by the constructions must also be \emph{perfect} and the solution obtained must be optimal for $(A,p')$.
  
\end{proof}

Let $p^*$ be one of the \emph{perfect} elements guaranteed to exist by Theorem~\ref{thrm:perfect_p} for $A$ and $S^*$ be a respective optimal solution for $(A,p^*)$ satisfying properties~\ref{part_prop1} and~\ref{part_prop2}. The following lemmas indicate that we can prune certain values of $p$, without missing $S^*$. We define $x=n-q$. Intuitively, for $(A,p^*)$, $x$ is the number of large elements that must be contained in singleton sets in $S^*$, according to Theorem~\ref{thrm:perfect_p}. 

\begin{lemma}
\label{lem:pigeonhole}
    If for a $\kPRRshort$ instance $(A,p)$ we have $x>k-1$, then $p \neq p^*$.
\end{lemma}

\begin{proof}
    Since $S_1$ cannot contain elements $j>q$, there are $k-1$ sets that can contain such elements. The number of these elements is $x$. If $x>k-1$, by the pigeonhole principle, every feasible solution contains a set with two or more of these elements. This contradicts property~\ref{part_prop2} of Theorem~\ref{thrm:perfect_p}, therefore $p \neq p^*$.
\end{proof}

\begin{lemma}
\label{lem:pigeonhole_special_case}
    If for a $\kPRRshort$ instance $(A,p)$ we have $x=k-1$ and $p < q$, then $p \neq p^*$.
\end{lemma}
\begin{proof}
    Because $p<q$, $S_1$ cannot contain $q$. Thus, there are at least $x+1=k$ elements that cannot be contained in $S_1$, with $k-1$ of them being greater than $q$. By the pigeonhole principle, every feasible solution contains a set with two or more elements, one of which is greater than $q$. This contradicts property~\ref{part_prop2} of Theorem~\ref{thrm:perfect_p}, therefore $p \neq p^*$.
\end{proof}

\subsection{Obtaining an FPTAS for \texorpdfstring{\textit{k}-PART}{k-PART}} 
\label{subsec:DP_partition}

We now present the main algorithm for $\kPRshort$. Its design is analogous to that of Algorithms $\ref{alg:exact_k-SSRR}$ and $\ref{alg:FPTAS_k-SSRR}$, with two important differences:

\begin{enumerate}
    \item There is no iteration for different amounts of singletons with elements $j>q$. Instead, we prune some values of $p$ according to Lemmas~\ref{lem:pigeonhole} and~\ref{lem:pigeonhole_special_case} and fix $x$ singletons for the rest.

    \item In contrast to $\kSSRRshort$, we do not necessarily obtain the optimal solution to each rounded $\kPRRshort$ instance $(A_r,p)$. Interestingly, we will later prove that the obtained solution is sufficient for the algorithm to be an FPTAS for $\kPRshort$.
\end{enumerate}

\begin{algorithm}[H] 
\caption{$\texttt{FPTAS\textunderscore $\kPRshort$}(A,\varepsilon)$}
\label{alg:FPTAS_k-PART}
\begin{algorithmic}

\Require{A sorted multiset $A=\{a_1,\ldots,a_n\}, a_i\in\mathbb{Z}^+$ and an error parameter $\varepsilon \in \left(0, 1\right)$.}

\Ensure{Disjoint subsets $(S_1,\ldots,S_k)$ of $[n]$, such that  $\bigcup_{i=1}^k S_i = [n]$ and their ratio satisfies $\mathcal{R}(S_1,\ldots,S_k,A) \leq (1+\varepsilon)\mathcal{R}(S_1^*,\dots,S_k^*,A)$.}

\For{$p\leftarrow1,\ldots,n-k+1$}
    \State $best\_ratio[p] \leftarrow \infty$,        
    $best\_solution[p] \leftarrow (\emptyset,\ldots,\emptyset)$
\EndFor

\For{$p \leftarrow 1,\ldots,n-k+1$}
    \State $Q \leftarrow \sum_{i=1}^p a_i$,       
    $q \leftarrow \max\{i\mid a_i\leq Q\}$, 
    $x \leftarrow n-q$

    \If{$x>k-1 \; \lor \; (x=k-1 \; \land \; p<q)$} \Comment{Lemmas~\ref{lem:pigeonhole} and~\ref{lem:pigeonhole_special_case}}
        \State \textbf{continue} to next $p$ 
    \EndIf 


    
    \State \textbf{for} $y\leftarrow1, \ldots, x$ \textbf{do}
        $S_{k-x+y} \leftarrow \{q+y\}$\Comment{$x$ singletons}

    \State $\delta \leftarrow \frac{\varepsilon \cdot a_p}{3\cdot n}$, $A_r\leftarrow \emptyset$, $k' \leftarrow k-x$

    \State \textbf{for} $i\leftarrow1,\ldots,n$ \textbf{do}
     $a^r_i \leftarrow \lfloor \frac{a_i}{\delta} \rfloor$,   
    $A_r\leftarrow A_r\cup \{ a^r_i \}$


    \State $A_r' \leftarrow \{a_1^r,\ldots,a_q^r\}$ \Comment{$q$ items for DP}
    \State $DP\_solutions \leftarrow$ \texttt{DP\textunderscore $\kPRRshort$}$(A_r',k',p,Q/\delta)$ \Comment{Call Alg.~\ref{alg:DP_k-PARR}}
    
    \For{\textnormal{\textbf{all}} $(S_1,\ldots,S_{k'})$ \textnormal{\textbf{in}} $DP\_solutions$}
            \State $current\_ratio \leftarrow \mathcal{R}(S_1,\ldots,S_k,A_r)$ \Comment{$\kPRRshort$ solution}

            \If{$current\_ratio < best\_ratio[p]$}
                \State $best\_solution[p] \leftarrow (S_1,\ldots,S_k)$ \Comment{Best solution for each $A_r$}
                \State $best\_ratio[p] \leftarrow current\_ratio$ \Comment{and its ratio}
            \EndIf
        \EndFor
\EndFor

\State $final\_ratio \leftarrow \infty, final\_solution \leftarrow 0$

\For{$p\leftarrow1,\ldots,n-k+1$} \Comment{Iterate for all $p$ to find best sol. for $A$}
    \State $current\_ratio \leftarrow \mathcal{R}(best\_solution[p],A)$
    
    \If{$current\_ratio < final\_ratio$}
        \State $final\_solution \leftarrow best\_solution[p]$
        \State $final\_ratio \leftarrow current\_ratio$
    \EndIf
\EndFor

\State \Return $final\_solution$

\end{algorithmic}
\end{algorithm}

Algorithm~\ref{alg:FPTAS_k-PART} calls a dynamic programming subroutine for each value of $p$, in order to find candidate partial solutions for elements $j \leq q$. This DP subroutine, namely Algorithm~\ref{alg:DP_k-PARR}, is a direct extension of Algorithm~\ref{alg:DP_k-SSRR}, with the following three differences.

\begin{enumerate}
    \item The case of an element not being added to any set is skipped.
    \item When a conflict occurs in a DP cell $T_i[D][V]$, we do not use $sum_1$ to resolve it. For $\kPRRshort$ every element is included in some set, thus conflicts can only occur between tuples with identical sums. Hence, it does not matter which tuple is preferred (as proven in Theorem~\ref{thrm:kSSRR_exact_correctness} for conflicts with equal sums) and there is no need to ever consider $sum_1$.
    
    \item The bound for pruning large negative differences is chosen as $-2Q/\delta$, where $Q=\sum_{i=1}^p a_i$ is calculated using the \emph{initial} values $a_i$ instead of the scaled and rounded values $a_i^r$. This is necessary to avoid pruning edge case solutions. We will expand on this in Lemma~\ref{lem:partition_almost_correctness}.
\end{enumerate}


\begin{algorithm}[H] 
\caption{$\texttt{DP\textunderscore $\kPRRshort$}(A,k,p,Q)$}
\label{alg:DP_k-PARR}
\begin{algorithmic}[1]

\Require{A sorted multiset $A=\{a_1,\ldots,a_q\}, a_i\in\mathbb{Z}^+$, an integer $k \geq 1$, an integer $p,1 \leq p \leq q-k+1$ and a real number $Q>0$.}

\Ensure{A set $solutions$ containing tuples of $k$ disjoint subsets $(S_1,\ldots,S_k)$ of $[q]$, such that $\bigcup_{i=1}^k S_i = [q]$, $\max(S_1)=p$, $\max(S_i)>p$ for $1<i \leq k$.}

\State $T_i[D][V] \leftarrow \emptyset$ \textbf{for all} $1 \leq i \leq q$, $D=[d_2,\ldots,d_k]$, $V=[v_2,\ldots,v_k]$ with $d_2 \leq \ldots \leq d_k$
\State $T_0[a_p,\ldots,a_p][\texttt{false},\ldots,\texttt{false}] = (\{p\}, \emptyset, \ldots, \emptyset)$

\For{$i \leftarrow 1,\ldots,q$}

     \If{$i=p$}
        \State $T_i[D][V] \leftarrow T_{i-1}[D][V]$ \textbf{for all} $T_{i-1}[D][V] \neq \emptyset$
        
        \State \textbf{continue} to next $i$
    \EndIf
    
    \For{\textnormal{\textbf{all}} $T_{i-1}[D][V] \neq \emptyset$}



        \If{$i<p$}
        \State $(S_1, \ldots, S_k), D', V' \leftarrow \texttt{update}(a_i,i,1,D,V,T_{i-1}[D][V])$
                
        \If{$T_{i}[D'][V'] = \emptyset$}
                \State $T_{i}[D'][V'] \leftarrow (S_1, \ldots, S_k)$ \EndIf
        \EndIf
        
        \For{$j\leftarrow2, \ldots, k$}
            \If{$d_j - a_i > -2Q$}
            
                \State $(S_1, \ldots, S_k), D', V' \leftarrow \texttt{update}(a_i,i,j,D,V,T_{i-1}[D][V])$
                
                \If{$T_{i}[D'][V'] = \emptyset$} 
                \State $T_{i}[D'][V'] \leftarrow (S_1, \ldots, S_k)$   
                \EndIf
            \EndIf
        \EndFor    
    \EndFor    
\EndFor

\State $solutions \leftarrow \emptyset$

\For{\textnormal{\textbf{all}} $T_q[D][\textnormal{\texttt{true},\ldots,\texttt{true}}] \neq \emptyset$}
    \State $(S_1, \ldots, S_k) \leftarrow T_q[D][\texttt{true},\ldots,\texttt{true}]$
        
    \State $solutions \leftarrow solutions \cup \{(S_1, \ldots, S_k)\}$
\EndFor

\State \Return $solutions$

\end{algorithmic}
\end{algorithm}

Algorithm~\ref{alg:DP_k-PARR} calls an \texttt{update} function, which is a simpler variation of Algorithm~\ref{alg:update}; essentially, it skips all modifications regarding $sum_1$. The pseudocode of this function is omitted due to its simplicity.

We now present the following lemma, whose proof is essentially identical to that of Lemma~\ref{lem:kSSRR_feasibility}.

\begin{lemma}[Feasibility]
\label{lem:partition_feasible}
    Every $k$-tuple of sets whose $\mathcal{R}(S_1,\ldots,S_k,A_r)$ value is considered by Algorithm $\ref{alg:FPTAS_k-PART}$ is a feasible solution for the $\kPRRshort$ instance $(A_r,p)$.
\end{lemma}

Recall that we defined $p^*$ and $S^*$ as a \emph{perfect} $p$ and a respective optimal solution for $(A,p^*)$ whose existence is guaranteed by Theorem~\ref{thrm:perfect_p}.

\begin{lemma}[Near-optimality]
\label{lem:partition_almost_correctness}
    When Algorithm $\ref{alg:FPTAS_k-PART}$ iterates to compare the ratio of solutions for the $\kPRRshort$ instance $(A_r,p^*)$, it will consider either $S^*$ or another solution $S=(S_1,\ldots,S_k)$ with $\Sigma(S_i,A_r)=\Sigma(S_i^*,A_r)$, $\forall i \in [k]$.
\end{lemma}

\begin{proof}
    According to Lemmas~\ref{lem:pigeonhole} and~\ref{lem:pigeonhole_special_case}, we cannot have $x>k-1$ or $x=k-1$ and $p<q$ for $p=p^*$. The singleton sets enforced by Algorithm $\ref{alg:FPTAS_k-PART}$ are exactly the singleton sets contained in $S^*$, according to property~\ref{part_prop2} of Theorem~\ref{thrm:perfect_p}.

    Recall that a solution in the DP subroutine is pruned if it has some difference $d_j \leq -2Q/\delta$. By property~\ref{part_prop1} of Theorem~\ref{thrm:perfect_p}, it holds that $\Sigma(S_i^*,A) < 2Q$, $\forall i \in [k-x]$. Since $A_r$ contains elements scaled by $\delta$ and \emph{rounded down}, we infer that $\forall i \in [k-x]$
    $$\Sigma(S_i^*, A_r) \leq \Sigma(S_i^*, A)/\delta <2Q/\delta.$$
    This implies that for $S^*$, there will be no $d_j \leq -2Q/\delta$ (at any point in its dynamic programming construction).

    As already explained, DP conflicts for $\kPRRshort$ occur only between tuples with identical sums. Recall that two conflicting tuples have the same vectors $D,V$ and only use elements up to some element $i$, so any combination of elements that can be added to sets of one tuple can also be added to the respective sets of the other tuple. This is explained in more detail in the proof of Theorem~\ref{thrm:kSSRR_exact_correctness}. We infer that it is possible for $S^*$ to be overwritten by another solution $S=(S_1,\ldots,S_k)$ with $\Sigma(S_i,A_r)=\Sigma(S_i^*,A_r)$, $\forall i \in [k]$.

    The DP subroutine constructs every possible combination of disjoint sets $S_1,\ldots,S_{k-x}$ with $\max (S_1) = p$ and $\bigcup_{i=1}^k S_i = [q]$, apart from the ones pruned by the cases mentioned in the previous paragraph. Thus, the lemma follows.
\end{proof}

\begin{theorem}
\label{thrm:partition_FPTAS}
    Algorithm $\ref{alg:FPTAS_k-PART}$ is an FPTAS for $\kPRshort$ that runs in time $O({n^{2k}}/{\varepsilon^{k-1}})$. 
\end{theorem}

\begin{proof}    
    We rely upon the observation that the proof of Theorem~\ref{thrm:approximation_ratio} does not necessarily require an optimal solution for the rounded $\kSSRRshort$ instance $A_r$; it suffices for the algorithm to consider a $\kSSRRshort$ solution $S$ with $\mathcal{R}(S,A_r) \leq \mathcal{R}(S_{opt},A_r)$, where $S_{opt}$ is an optimal solution for the $\kSSRshort$ instance $A$ (prior to rounding).

    First, consider the case $p\not=p^*$. If the conditions of Lemma~\ref{lem:pigeonhole} or~\ref{lem:pigeonhole_special_case} are met, this $p$ will be skipped. Otherwise, a $k$-tuple of sets will be found, which is a feasible solution for the $\kPRRshort$ instance $(A_r,p)$, according to Lemma~\ref{lem:partition_feasible}.
    
    Second, consider the case $p=p^*$. By Lemma~\ref{lem:partition_almost_correctness}, Algorithm~\ref{alg:FPTAS_k-PART} will consider $S^*$ or another solution $S$ with $\mathcal{R}(S,A_r)=\mathcal{R}(S^*,A_r)$ as a solution for the instance $(A_r,p^*)$. Therefore, for the solution $S_{alg}$ found by the algorithm in this iteration, it holds that $\mathcal{R}(S_{alg},A_r) \leq \mathcal{R}(S^*,A_r)$. With the aforementioned observation, the proof of the following inequality is identical to the proof of Theorem~\ref{thrm:approximation_ratio}.
    $$\mathcal{R}(S_{alg},A) \leq (1+\varepsilon)\mathcal{R}(S^*,A)$$

    Taking into account both cases, the final solution returned by the algorithm is a feasible $\kPRshort$ solution for $A$ and has ratio smaller than or equal to that of the solution found in the $p^*$-th iteration. It follows that Algorithm~\ref{alg:FPTAS_k-PART} is a $(1+\varepsilon)$-approximation for $\kPRshort$.

    We now analyze the complexity of the algorithm. The DP subroutine runs in time $O(n(Q/\delta)^{k-1})$, by the same reasoning as in the proof of Lemma~\ref{lem:kSSRR_exact_complexity} for $x=0$ (which is the worst case). Taking into account the iteration for all possible values of $p$ and using a bound for $Q/\delta$ just like we did for $\kSSRshort$, we infer that Algorithm~\ref{alg:FPTAS_k-PART} runs in $O({n^{2k}}/{\varepsilon^{k-1}})$.
\end{proof}

\section{An FPTAS for \texorpdfstring{\textit{k}-SSR}{k-SSR} with linear dependence on \texorpdfstring{\textit{n}}{n}}
\label{sec:FPTAS_kSSR_alternative}

In this section, we provide an alternative FPTAS for $\kSSRshort$, using techniques inspired from \cite{SSR_Bri} to reduce the problem to smaller instances. We then use the $\kSSRshort$ FPTAS of Section~\ref{sec:FPTAS_kSSR} to solve them. 

We define an alternative restricted version of $\kSSRshort$, called $\kSSRLshort$, in which the solution is forced to contain the largest element of the input. This definition is based on the $\mathrm{SSR_L}$ problem defined in \cite{SSR_Bri}.

\begin{definition}[$\kSSRLshort$]
\label{def:kSSRL}
Given a sorted multiset $A=\{a_1,\ldots,a_n\}$ of positive integers, find $k$ disjoint subsets $S_1,\ldots,S_k$ of $[n]$ with $n \in \bigcup_{i=1}^k S_i$, such that $\mathcal{R}(S_1,\ldots,S_k,A)$ is minimized.
\end{definition}

Let $A=\{a_1,\ldots,a_n\}$ be a sorted multiset of $n$ positive integers. We denote by $A[l,r]$ the subset of $A$ consisting of all items $a_i$, such that $l\leq i \leq r$. We also use the notation $\textup{sum} (A)=\sum_{a \in A} a$.

For a multiset $A=\{a_1,\ldots,a_n\}$, we denote the optimal ratio of the $\kSSRshort$ instance $A$ as $\mathrm{OPT}(A)$ and the optimal ratio of the $\kSSRLshort$ instance $A$ as $\mathrm{OPT_L}(A)$. 
By definition, it holds that
$1\leq \mathrm{OPT}(A) \leq \mathrm{OPT_L}(A)$.

\begin{lemma}
\label{lem:optl_min}
    Let $A=\{a_1,\ldots,a_n\}$ be a sorted multiset of positive integers. It holds that
    $$\mathrm{OPT}(A)=\min_{k\leq j \leq n}\mathrm{OPT_L}\left(A[1,j]\right).$$
\end{lemma}

\begin{proof}
    Let $t\in\{k,\ldots,n\}$ be the maximum element\footnote{If $t<k$, all feasible solutions for both problems have infinite ratio.} contained in some optimal solution $S$ of the $\kSSRshort$ instance $A$. By Definition~\ref{def:kSSRL}, $S$ is a feasible solution for the $\kSSRLshort$ instance $A[1,t]$. This implies $\mathrm{OPT}(A)\geq\min_{k\leq j \leq n}\mathrm{OPT_L}\left(A[1,j]\right)$.

    Note that any feasible solution for a $\kSSRLshort$ instance $A[1,j]$ (for any $j\in\{k,\ldots,n\}$) is a feasible solution for the $\kSSRshort$ instance $A$. Thus, it also holds that $\mathrm{OPT}(A)\leq\min_{k\leq j \leq n}\mathrm{OPT_L}\left(A[1,j]\right)$.
\end{proof}

Lemma~\ref{lem:linear_inequality} provides the key structural insight of our algorithm. Intuitively, the fact that a $\kSSRLshort$ solution contains the largest element allows us to consider only the largest $O\bigl((1/\varepsilon)\ln(1/\varepsilon)\bigr)$ elements for each $\kSSRLshort$ instance $A[1,j]$, in order to obtain a $(1+\varepsilon)$-approximation of the optimal $\kSSRshort$ solution. 

\begin{lemma}[Largest elements]
    \label{lem:linear_inequality}
    For any sorted multiset $A=\{a_1,\ldots,a_n\}$ of positive integers and any $\varepsilon \in (0,1)$, the following holds:
    \begin{equation}\label{largest_elements_ineq}
        \mathrm{OPT}(A) \leq \min_{k\leq j \leq n}\mathrm{OPT_L}\left(A\bigl[j-C+1,j\bigr]\right) \leq (1+\varepsilon)\mathrm{OPT}(A),
    \end{equation}
    where $C=(c+1)(k-1)$ and $c= 1+\left\lceil\left(1+\frac1\varepsilon \right)\ln\frac{2(k-1)}{\varepsilon^2}\right\rceil$.
\end{lemma}

\begin{proof}
    Since $A\bigl[j-C+1,j\bigr]$ is a subset of $A$ and $\kSSRLshort$ is (by definition) a restricted version of $\kSSRshort$, it follows directly that
    \begin{equation*}
    \mathrm{OPT}(A) \leq \mathrm{OPT}\left(A\bigl[j-C+1,j\bigr]\right)
     \leq \mathrm{OPT_L}\left(A\bigl[j-C+1,j\bigr]\right).
    \end{equation*}
    Note that the above is true for any $j \in [n]$, therefore it is also true for the minimum. Thus, the first inequality of \eqref{largest_elements_ineq} holds. As for the second inequality of \eqref{largest_elements_ineq}, we distinguish between two cases.


    \emph{Case 1 (Dense case):} There exists some $v \in \{k,\dots,n\}$ s.t. $a_{v} \leq
    (1+\varepsilon)a_{v-k+1}$.  Since $C\ge k$, by setting $S_i = \{v-k+i\}$ for each $i\in[k]$, we obtain
    \begin{equation*}
    \min_{k\leq j \leq n} \mathrm{OPT_L}(A[j-C+1,j]) \leq \mathrm{OPT_L}(A[v-C+1,v]) \leq (1+\varepsilon) \leq (1+\varepsilon)\mathrm{OPT}(A)
    .\end{equation*}

    \emph{Case 2 (Sparse case):} For all $v \in \{k,\dots,n\}$ it holds that $a_v > (1+\varepsilon)a_{v-k+1}$. By repeatedly applying this inequality, it follows that for all $v\in \{k,\dots,n\}$ and all $i\geq 1$ such that $v-i(k-1)\geq 1$,
    \begin{equation}
    \label{sparse_ineq}
        a_{v-i(k-1)}<(1+\varepsilon)^{-i}a_v.
    \end{equation}
    Thus, for all such $v,i$ it holds that
\begin{align}
\begin{split}
    \textup{sum} \Bigl( A\bigl[v-(i+1)(k-1)+1,\  v-i(k-1)\bigr]\Bigr) &\leq (k-1)a_{v-i(k-1)}\\
    &<(k-1)(1+\varepsilon)^{-i}a_v. 
    \quad \mbox{[by Ineq.~}\eqref{sparse_ineq}\mbox{]}
    \label{interval_ineq}
\end{split}
\end{align}

    Consider the set $A\bigl[1,v-c(k-1)\bigr]$. Splitting this set into subsets of size (at most) $k-1$ and using \eqref{interval_ineq} for each of them yields
    \begin{equation}
    \label{geometric_series}
    \textup{sum} \Bigl( A\bigl[1,v-c(k-1)\bigr]\Bigr) \leq (k-1)\sum_{i=0}^{\infty}(1+\varepsilon)^{-(c+i)}a_v = \frac{(k-1)a_v}{\varepsilon(1+\varepsilon)^{c-1}},
    \end{equation}
    where the last step is calculated as the sum of a geometric series. Note that \eqref{geometric_series} holds even if $v-c(k-1)<1$, since then it would be $\textup{sum} \Bigl( A\bigl[1,v-c(k-1)\bigr]\Bigr)=0$.
    Substituting $c$ into \eqref{geometric_series} and using $(1+\varepsilon)^{1+1/\varepsilon}>e$ yields
    \begin{equation*}
        \textup{sum} \Bigl( A\bigl[1,v-c(k-1)\bigr]\Bigr)<\frac\varepsilon2\cdot a_v.
    \end{equation*}
    For any $k\leq j \leq n$, we choose $v=j-k+1$ to obtain

    \begin{equation}
    \label{final_inequality}
        \textup{sum} \Bigl( A\bigl[1,j-C\bigr]\Bigr)<\frac\varepsilon2\cdot a_{j-k+1}.
    \end{equation}

    Intuitively, inequality \eqref{final_inequality} is a bound for the sum of the $j-C$ \emph{smallest} elements of $A$. This inequality will be used to obtain a $(1+\varepsilon)$-approximation by removing these elements from an optimal $\kSSRLshort$ solution. 

    Let $S=(S_1,\dots,S_k)$ be an optimal solution for the $\kSSRLshort$ instance $A_j=A[1,j]$. We define
$$
S_M = \argmax_{S_i \in S} \Sigma(S_i,A_j)
\quad\mbox{and}\quad
S_m =  \argmin_{S_i \in S}\Sigma(S_i,A_j).
$$

Now consider the $\kSSRLshort$ instance $A_j' = A\bigl[j-C+1,j\bigr]$ (i.e.\ $A_j'$ contains the $C$ \emph{largest} elements of $A_j$). For each $S_i \in S$, we define $S'_i = S_i \cap [j-C+1,j]$, i.e. the subset of $S_i$ that contains elements $i$ such that $a_i \in A_j'$. We also define the respective $k$-tuple of sets, $S' = (S_1',\dots,S_k')$, and the following sets,
$$
S_{\mathcal{M}} = \argmax_{S_i \in S} \Sigma(S'_i,A_j')
\quad \textnormal{and} \quad
S_\mu =  \argmin_{S_i \in S}\Sigma(S'_i,A_j').
$$
Observe that, by the above definitions, $S'_{\mathcal{M}}$ is a set with maximum sum among all sets in $S'$ and $S'_\mu$ is a set with minimum sum among all sets in $S'$. By \eqref{final_inequality}, for all $i \in [k]$ it holds that $$\Sigma(S_i,A_j) - \Sigma(S'_i,A_j') \leq \textup{sum} \Bigl( A\bigl[1,j-C\bigr]\Bigr)<\frac{\varepsilon}{2}\cdot a_{j-k+1}.$$
Therefore, we obtain
\begin{equation} \label{results_of_trimming}
    \Sigma(S'_i,A_j') \leq \Sigma(S_i,A_j) \leq \Sigma(S'_i,A_j') + (\varepsilon/2)\cdot a_{j-k+1}, \; \forall i\in [k].
\end{equation}
We now obtain a bound for the ratio $\mathcal{R}(S',A_j')$ as follows.
\begin{align}
\begin{split}
\label{ratio_inequality_pro_max}
    \mathcal{R}(S',A_j') &= \frac{\Sigma\left(S'_{\mathcal{M}}, A_j'\right)}{\Sigma\left(S'_\mu,A_j' \right)} \leq \frac{\Sigma\left(S_{\mathcal{M}}, A_j\right)}{\Sigma\left(S_\mu, A_j\right) - \frac{\varepsilon}{2} \cdot a_{j-k+1}} \qquad \mbox{[by Ineq.~}\eqref{results_of_trimming}\mbox{]}\\
    &\leq \frac{\Sigma\left(S_M, A_j\right)}{\Sigma\left(S_m, A_j\right) - \frac{\varepsilon}{2} \cdot a_{j-k+1}}.
\end{split}
\end{align}

Recall that $j \in \bigcup_{i=1}^k S_i$ (by Definition~\ref{def:kSSRL}), which implies that $\Sigma(S_M,A_j) \geq a_j$. Let $S^s=\{S^s_1,\ldots,S^s_k\}$ be a solution consisting of $k$ singleton sets that contain the $k$ largest elements of $A_j$, i.e. $S^s_i = \{j-i+1\}$, $\forall i\in[k]$. Since $S^s$ is a feasible solution for $\kSSRLshort$ with ratio $a_{j}/a_{j-k+1}$, it holds that $\mathcal{R}(S,A_j)=\Sigma(S_M, A_j)/\Sigma(S_m, A_j) \leq a_{j}/a_{j-k+1}$. Thus, we have
$$\Sigma(S_m, A_j) \geq \Sigma(S_M, A_j) \cdot \frac{a_{j-k+1}}{a_{j}} \geq a_{j-k+1}.$$
Combining this with \eqref{ratio_inequality_pro_max}, we obtain 
$$\mathcal{R}(S',A_j') \leq \frac{\Sigma(S_M, A_j)}{\Sigma(S_m, A_j)\left(1-\frac{\varepsilon}{2}\right)} =\frac{2}{2-\varepsilon} \cdot \mathcal{R}(S,A_j).$$
Assuming $\varepsilon \in (0,1)$, this becomes $\mathcal{R}(S',A_j') \leq (1+\varepsilon) \mathcal{R}(S,A_j)$. Note that some set $S'_i$ contains $j$, so $S'$ is a feasible solution for the $\kSSRLshort$ instance $A_j'$, therefore
$$\mathrm{OPT_L}(A_j') \leq \mathcal{R}(S',A_j') \leq (1+\varepsilon) \mathcal{R}(S,A_j) = (1+\varepsilon)\mathrm{OPT_L}(A_j).$$
By the definitions of $A_j,A_j'$, we have the following for all $j \in \{k,\ldots,n\}$:
$$\mathrm{OPT_L}\left(A\bigl[j-C+1,j\bigr]\right) \leq  (1+\varepsilon)\mathrm{OPT_L}(A[1,j]).$$
Taking the minimum over all $j\in\{k,\ldots,n\}$ and using Lemma~\ref{lem:optl_min}, we have

$$\min_{k\leq j \leq n}\mathrm{OPT_L}\left(A\bigl[j-C+1,j\bigr]\right) \leq  (1+\varepsilon)\mathrm{OPT}(A).$$

\end{proof}

\begin{lemma}[Reduction]
\label{lem:kssrL_reduction}
    If there is a $(1+\varepsilon)$-approximation algorithm for $\kSSRLshort$ running in time $T_L(n,\varepsilon)$, then there is a $(1+\varepsilon)$-approximation algorithm for $\kSSRshort$ running in time
    $$
        O\left(nT_L\left(\frac{9k}{\varepsilon}\ln\frac{k}{\varepsilon},\ \frac\varepsilon3\right)\right).
    $$
\end{lemma}

\begin{proof}
    Let $A$ be a multiset of $n$ positive integers $a_1 \leq \ldots \leq a_n$. For each $j \in \{k,\dots,n\}$, consider $A_j'=A\bigl[j-C+1,j\bigr]$, where
    $$C=(c+1)(k-1) \qquad \textnormal{and} \qquad c= 1+\left\lceil\left(1+\frac1\varepsilon \right)\ln\frac{2(k-1)}{\varepsilon^2}\right\rceil.$$
    
    Using the given $(1+\varepsilon)$-approximation algorithm for each of the $\kSSRLshort$ instances $A_j'$, we receive a $k$-tuple of sets $S^j = (S^j_1,\dots, S^j_k)$ for which $\mathcal{R}\left(S^j,A_j'\right) \leq (1+\varepsilon)\mathrm{OPT_L}\left(A_j'\right)$. Thus,

    \begin{align*}
        \min_{k\leq j \leq n} \mathcal{R}\left(S^j,A_j'\right) \leq (1+\varepsilon)\min_{k\leq j \leq n}\mathrm{OPT_L}\left(A_j'\right)\leq (1+\varepsilon)^2\mathrm{OPT}(A). &&\mbox{[by Lemma~}\ref{lem:linear_inequality}\mbox{]}
    \end{align*}
    
    For all $j \in \{k,\ldots,n\}$, the set $A_j'$ contains at most $C$ elements, with
    \begin{align*}
    C=(k-1)\left(2+\left\lceil\left(1+\frac1\varepsilon \right)\ln\frac{2(k-1)}{\varepsilon^2}\right\rceil\right) <
    3k + k\left(1+\frac1\varepsilon \right)\ln\frac{2k}{\varepsilon^2}
    < \frac{9k}{\varepsilon}\ln\frac{k}{\varepsilon}
    ,\end{align*}
    assuming $k\geq 2$ and $\varepsilon \in (0,1)$. As such, we have to run the given algorithm on $O(n)$ $\kSSRLshort$ instances of size bounded by $(9k/\varepsilon)\ln(k/\varepsilon)$ in order to obtain a $(1+\varepsilon)^2$-approximation for $\mathrm{OPT}(A)$. For all $\varepsilon \in (0,1)$ it holds that $(1+\varepsilon)^2 \leq 1+3\varepsilon$. Therefore, by setting the error parameter to $\varepsilon/3$ instead of $\varepsilon$ in the approximation algorithm for $\kSSRLshort$, we obtain a $(1+\varepsilon)$-approximation for $\mathrm{OPT}(A)$ running in time
    \begin{equation*}
         O\left(nT_L\left(\frac{9k}{\varepsilon}\ln\frac{k}{\varepsilon},\ \frac\varepsilon3\right)\right).
    \end{equation*}
\end{proof}

We now present the main theorem of this section, which follows by using the FPTAS of Theorem~\ref{thrm:main} as the $(1+\varepsilon)$-approximation algorithm for $\kSSRLshort$ in Lemma~\ref{lem:kssrL_reduction}.

\begin{theorem}
\label{thrm:linear_fptas}
    There is an FPTAS for $\kSSRshort$ that runs in $\widetilde{O}(n/{\varepsilon^{3k-1}})$ time, where $\widetilde{O}$ hides $\mathrm{polylog}(1/\varepsilon)$ factors.
\end{theorem}

\begin{proof}
    Let $A$ be a $\kSSRshort$ instance. Suppose we use our $\kSSRshort$ FPTAS running in time $T_L(n,\varepsilon)=O({n^{2k}}/{\varepsilon^{k-1}})$ (see Theorem~\ref{thrm:main}) as a subroutine to solve the $\kSSRLshort$ instances $A_j'=A\bigl[j-C+1,j\bigr]$, $k\leq j\leq n$. This subroutine serves as a $(1+\varepsilon)$-approximation for $\kSSRLshort$, since for the solution $S^j = (S^j_1,\dots, S^j_k)$ returned it holds that
    $$\mathcal{R}\left(S^j,A_j'\right)\leq (1+\varepsilon)\mathrm{OPT}(A_j') \leq (1+\varepsilon)\mathrm{OPT_L}(A_j').$$

    Note that this subroutine does not take into account the $\kSSRLshort$ restriction $j \in \bigcup_{i=1}^k S_i^j$, therefore it might return an invalid $\kSSRLshort$ solution for $A_j'$. However, all solutions returned are feasible for the (unrestricted) $\kSSRshort$ instance $A$, therefore the solution which yields $\min_{k\leq j \leq n} \mathcal{R}\left(S^j,A_j'\right)$ is a feasible one. As such, by Lemma~\ref{lem:kssrL_reduction} we obtain another FPTAS for $\kSSRshort$ running in time

    $$O\left(nT_L\left(\frac{9k}{\varepsilon}\ln\frac{k}{\varepsilon},\ \frac\varepsilon3\right)\right)=\widetilde{O}\left(\frac n {\varepsilon^{3k-1}}\right).$$
\end{proof}

\section{Conclusion}
\label{sec:conclusion}

In this work we presented the first FPTASs for \kSSRshort\ with constant $k>2$. A natural open question that arises regards the existence of a PTAS for \kSSRshort\ when $k$ is a variable; this would require different techniques, since both of our \kSSRshort\ FPTASs contain factors exponential in $k$ in their time complexity. Note that such a PTAS already exists for \kPRshort\ by Lipton et al.~\cite{FD_LM}. 

We remark that our FPTAS in Section~\ref{sec:FPTAS_kSSR_alternative} is only faster than the one in Section~\ref{sec:FPTAS_kSSR} when $n \gg 1/\varepsilon$. This contrasts the current status of the FPTASs for $k=2$, where the state-of-the-art FPTAS by Bringmann~\cite{SSR_Bri} is the fastest without any tradeoff. This naturally leads to an open question regarding the existence of a stronger density argument that works for $k>2$ and is always faster than the DP-based approach.

Another interesting direction for future work would be to extend the methods of Alonistiotis et al.~\cite{SSR_Alo} for $k>2$, possibly leading to an improvement over the approximation schemes we presented in this work.
Even though the FPTASs presented by Alonistiotis et al.~\cite{SSR_Alo} are slower than Bringmann's FPTASs~\cite{SSR_Bri}, their techniques might prove useful if generalized to $k>2$, since Bringmann's density arguments do not seem to generalize well for $k>2$ (cf. Subsection~\ref{subsec:previous_work_SSR}).

Regarding $\kPRshort$, a meaningful direction involves extending our techniques to the more general \textsc{Minimum Envy-Ratio} problem, involving general additive valuations or at least some sub-class such as \emph{ordered} additive valuations. This could potentially achieve a better time complexity than that of Nguyen and Rothe's FPTAS~\cite{FD_NR} for \textsc{Minimum Envy-Ratio}, which works for general additive valuations.

Lastly, as stated in~\cite{SS_Bri}, it is unlikely for the current state-of-the-art FPTAS for SSR to be optimal. We restate the open problem of determining the optimal constant $c \geq 0$ such that SSR can be solved in time $O(n/\varepsilon^c)$ and extend this challenge to $\kSSRshort$. More generally, the development of lower bounds for \kSSRshort\ and \kPRshort\ would be particularly beneficial.

\subsubsection*{Acknowledgements}

This work has been partially supported by project MIS 5154714 of the National Recovery and Resilience Plan Greece 2.0 funded by the European Union under the NextGenerationEU Program.

No artificial intelligence was used at any point during the production of this paper. All of the presented results were conceived exclusively by the authors and the entirety of the paper was written exclusively by the authors.





\bibliography{./bibliography}

\end{document}